\documentclass[10pt,final]{elsarticle}
\pdfoutput=1
\usepackage{amsmath}
\usepackage{amssymb}
\usepackage{graphicx}
\usepackage{subfigure}
\usepackage{epstopdf}

\def\R{\mathbb R}

\def\x{\mathbf x}
\def\w{\omega}

\def\id{\textrm{Id}}

\newtheorem{theorem}{Theorem}
\newtheorem{lemma}{Lemma}

\newdefinition{rmk}{Remark}
\newdefinition{example}{Example}
\newproof{proof}{Proof}
\newproof{pot}{Proof of Theorem \ref{th:main}}

\begin{document}

\begin{frontmatter}

\title{Unstable gap solitons in inhomogeneous Schr\"odinger equations}

\author[label4]{R.\ Marangell}
\author[label2]{H.\ Susanto}
\author[label1]{C.K.R.T.\ Jones}
\address[label1]{Department of Mathematics, University of North Carolina at Chapel Hill,\\ Chapel Hill, NC 27599}
\address[label2]{School of Mathematical Sciences, University of Nottingham, University Park,\\ Nottingham NG7 2RD, UK}
\address[label4]{Department of Mathematics and Statistics, University of Sydney, \\ Sydney, NSW 2006 Australia}

\begin{abstract}
A periodically inhomogeneous Schr\"odinger equation is considered. The inhomogeneity is reflected through a non-uniform coefficient of the linear and non-linear term in the equation. Due to the periodic inhomogeneity of the linear term, the system may admit spectral bands. When the oscillation frequency of a localized solution resides in one of the finite band gaps, the solution is a gap soliton, characterized by the presence of infinitely many zeros in the spatial profile of the soliton. Recently, how to construct such gap solitons through a composite phase portrait is shown. By exploiting the phase-space method and combining it with the application of a topological argument
, it is shown that the instability of a gap soliton can be described by the phase portrait of the solution. Surface gap solitons 
at the interface between a periodic inhomogeneous and a homogeneous medium are also discussed. Numerical calculations are presented accompanying the analytical results. 
\end{abstract}

\end{frontmatter}

\section{Introduction}

A homogeneous nonlinear system may admit a localized solutions with a natural frequency residing in the first (semi-infinite) band-gap of the corresponding linear system. When there is a periodic non-uniformity in the linear system, additional finite band-gaps will be formed and the nonlinear system will admit a novel type of solitons known as the gap solitons \cite{denz09}. One main characteristic of a gap soliton is the infinitely many number of zeros in the profile of the solution, inheriting a characteristic of Bloch waves. Gap solitons are intensively studied among others in nonlinear optics \cite{acev00} and Bose-Einstein condensates \cite{kevr08}. Several reports on the experimental observation of gap solitons in the fields in the one-dimensional setting are, e.g., \cite{chen87,eggl96,eier04,mand03,mand04,rosb06}.

Depending on particular underlying assumptions and specific limits, gap solitons have been studied analytically through several different approaches. The first theoretical approach is through the coupled-mode theory, which is based on a decomposition of the wave field into a forward and backward propagating wave \cite{volo81,chen87,ster94}. The applicability and justification of the method can be seen in \cite{good01,peli07,peli08}. The stability of gap solitons in this approach have been studied analytically in \cite{malo94,ross98,bara98}. The second formal approximation to gap soliton is through the so-called tight-binding approximation, which leads to a discrete nonlinear Schr\"odinger equation (DNLS) \cite{kevr09}. In this approach, a gap soliton can be related to the 'ordinary' soliton through the so-called staggering transformation. The existence and the stability of discrete solitons in the uncoupled limit of this approach has been discussed in \cite{peli05}. The third analysis of gap solitons is based on the approximation when the eigenfrequency of the localized modes is close to one of the edges of the finite band gaps \cite{iizu94,iizu97,peli04,yang10}. In this case, the envelope of the gap solitons is described by the nonlinear Schr\"odinger equation. It is shown in \cite{peli04} that gap solitons at least suffer from an oscillatory instability because gap solitons possess internal modes.

Relatively recently, another analytical method was proposed by Kominis et al.\ \cite{komi06,komi06_2,komi07}, employing a phase space method for the construction of an analytical solitary wave. Even though the method is rather limited to piecewise-constant coefficients, it was shown that the method is effective in obtaining various types of localized modes belonging to gap solitons. For that new method, the stability result was so far only obtained through numerical simulations.

The phase-space method proposed in \cite{komi06,komi06_2,komi07} is similar to that used in our recent work \cite{rmckrtjhs10}, where it was shown that the profile of a solution in the phase-space can be used to describe its instability. The method was based on the topological argument developed in \cite{ckrtj88}. Here, we propose to apply a similar method to determine the stability of gap solitons obtained through the phase-space method \cite{komi06,komi06_2,komi07}. Despite the similarity in the proposed method in investigating the instability of gap solitons, the problem is nontrivial. The topological argument in \cite{ckrtj88} is so far immediately applicable to nonlinear systems with finite inhomogeneity (see \cite{rmckrtjhs10} and references therein). 
By specifically constructing the solutions, we show that the argument is also useful to study gap solitons. In addition to inhomogeneities occupying the infinite domain, the so-called surface gap solitons sitting at the interface between inhomogeneities in the semi-infinite domain and a homogeneous region \cite{rosb06,smir06,komi07} will also be studied. Our result will complement the numerical results on the stability of surface gap solitons recently studied, e.g., in \cite{dohn08,blan11}.

The paper is outlined as the following. In Section 2, the governing equations are discussed and the corresponding linear eigenvalue problem is derived. The construction of gap solitons using the phase-space method is briefly explained. The instability of gap solitons is studied analytically in Section 3 using the topological argument. In Section 4, the linear eigenvalue problem for several gap solitons is solved numerically, where an agreement between the analytical results presented in the previous section is obtained. In the same section, the instability of surface gap solitons is also discussed. We conclude the paper in Section 5.

\section{Mathematical model}

We consider the following governing system of differential equations
\begin{equation}\begin{array}{lll}
i\Psi_{t}+\Psi_{xx}+|\Psi|^2\Psi=V\Psi && x \in U_O := \R \setminus U_I\\
i\Psi_{t}+\Psi_{xx}-\eta |\Psi|^2\Psi=0 && x \in U_I
\end{array}
\label{gov1}
\end{equation}
where the `outer' equation has focusing type nonlinearity, the `inner' equation can be defocusing ($\eta>0$) or linear ($\eta=0$), and $U_O, U_I$ are disjoint sets of intervals to be specified later.

To study standing waves of (\ref{gov1}), we pass to a rotating
frame and consider solutions of the form $\Psi(x,t) = e^{-i \w t} \psi(x,t)$. We then have
\begin{equation}\begin{array}{lll}
i\psi_{t}+\psi_{xx}+|\psi|^2\psi=(V-\w) \psi&& x \in U_O,\\
i\psi_{t}+\psi_{xx}-\eta |\psi|^2\psi=-\w \psi && x \in U_I.
\end{array}
\label{gov2}
\end{equation}
Standing wave solutions of (\ref{gov1}) will be steady-state solutions to (\ref{gov2}). 
We consider real, $t$ independent solutions $u(x)$ to the ODE:
\begin{equation}\begin{array}{ccccc}
u_{xx} &=& (V-\w )u - u^3 & & x \in U_O, \\
u_{xx} &=& -\w u + \eta u^3 & & x \in U_I.
\end{array} \label{stat1}\end{equation}
\noindent To obtain solutions that decay to 0 as $x \to \pm \infty$, the condition that $V-\w > 0$ is required, with $\w \in \R_+$. We will also require that $u_x \to 0$ as $x \to \pm \infty$. To establish the instability of a standing wave solution we linearize (\ref{gov2}) about a solution to (\ref{stat1}). Writing $\psi=u(x) + \epsilon\left((r(x)+is(x))e^{\lambda t} + (r(x)^\star+is(x)^\star) e^{\lambda^\star t}\right)$ and retaining terms linear in $\epsilon$ leads to the eigenvalue problem
\begin{equation}
\lambda\left(\begin{array}{cc} r \\ s \end{array}\right) = \left(\begin{array}{cc} 0 & D_{-} \\ - D_{+} & 0 \end{array}\right) \left(\begin{array}{cc} r \\ s \end{array}\right)  = M \left(\begin{array}{cc} r \\ s \end{array}\right),
\label{eq:linear}
\end{equation}
where the linear operators $D_{+}$ and $D_{-}$ are defined as
\begin{eqnarray}
\begin{array}{lll}
D_{+} = \begin{array}{lll}
   \frac{\partial^2}{\partial x^2} - (V- \w ) + 3u^2, & x \in U_O, \\
   \frac{\partial^2}{\partial x^2} + \w - 3\eta u^2, & x \in U_I,
  \end{array}
\end{array}
\label{eq:plusoperator}\\
\begin{array}{lll}
D_{-} = \begin{array}{lll}
   \frac{\partial^2}{\partial x^2} - (V- \w ) + u^2, & x \in U_O, \\
   \frac{\partial^2}{\partial x^2} + \w - \eta u^2, & x \in U_I.
  \end{array}
\end{array}
\label{eq:minusoperator}
\end{eqnarray}
It is then clear that the presence of an eigenvalue of $M$ with positive real part implies instability.

In \cite{komi06} a gap soliton was constructed via a method of superimposing the phase portraits of the `outer' system:
\begin{equation}
u_x = y, \qquad
y_x = (V-\omega)u-u^3,
\label{eq:outer}
\end{equation}
and the `inner' one:
\begin{equation}
u_x = y, \qquad
y_x = -\omega u+\eta u^3.
\label{eq:inner}
\end{equation}
We can view the composite picture as a single, non-autonomous system with phase plane given by:
\begin{equation}
\begin{array}{lll}
u_x = y, \\
y_x = \left\{\begin{array}{lll}
   (V-\omega)u-u^3, & x \in U_O, \\
   -\omega u+\eta u^3, &x \in U_I.
  \end{array}\right.
\end{array}
\label{fiberdyn}
\end{equation}

In the phase plane of (\ref{eq:outer}), the outer system admits a soliton solution, given by the equation:
\begin{equation} \label{eq:homo}
y^2 = (V - \w) u^2 - \frac{u^4}{2},
\end{equation}
while solution curves of the inner system are given by
\begin{equation} \label{eq:innerC}
y^2 = - \w u^2 + \frac{\eta u^4}{2} + C.
\end{equation}
The inner system (\ref{eq:inner}) admits a heteroclinic orbit in the phase plane given by $C=\w^2/2$.
The solutions we are interested in will travel in the phase plane along the homoclinic orbit of the outer system described by (\ref
{eq:homo}) and then `flip' to the inner system as $x$ passes through $U_I$, and then `flip' back to the outer system along the homoclinic orbit, repeating the process for each of the components of $U_I$ (see \cite{komi06}).

\begin{figure}
\subfigure[]{\includegraphics[scale=.4]{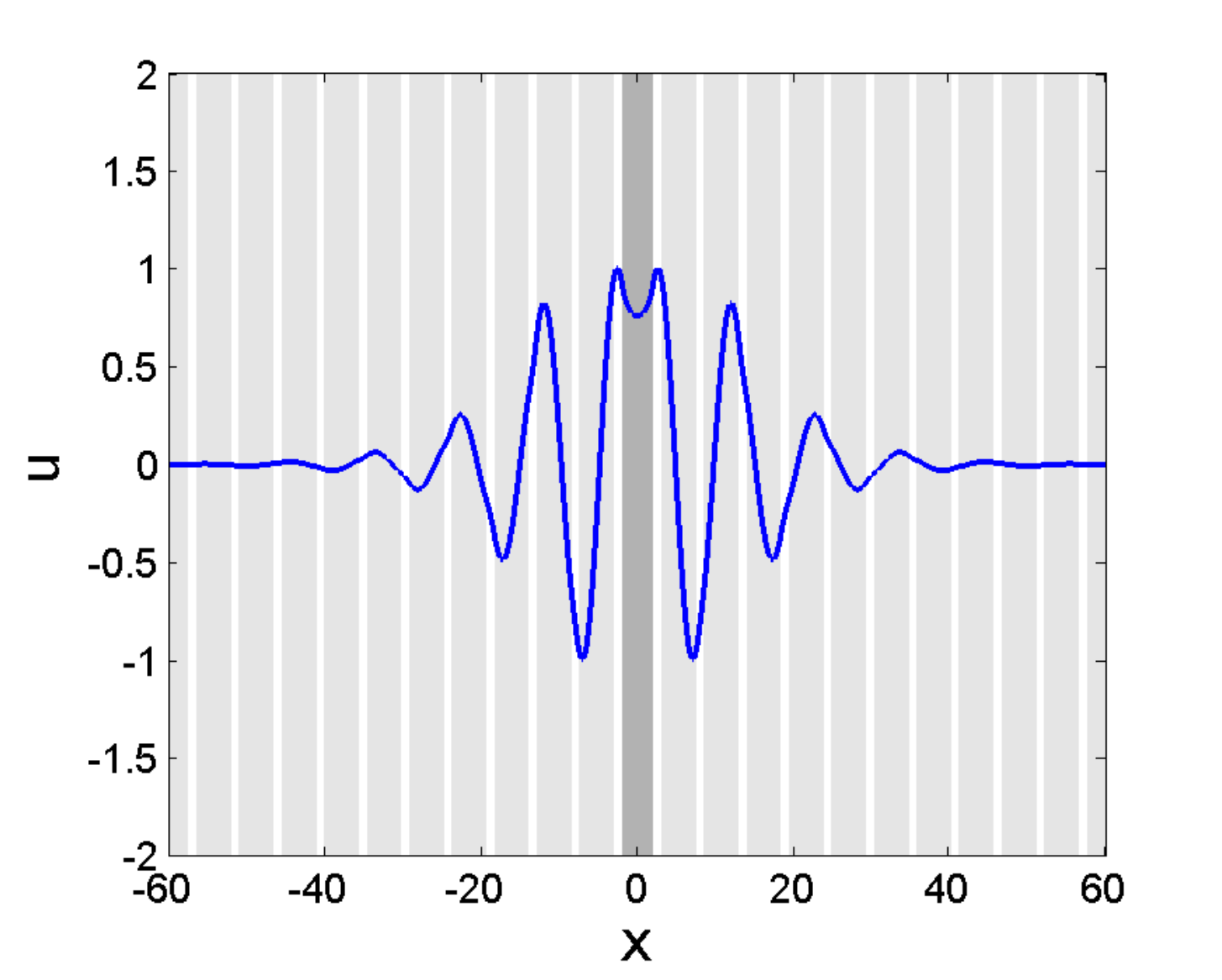}}
\subfigure[]{\includegraphics[scale=.4]{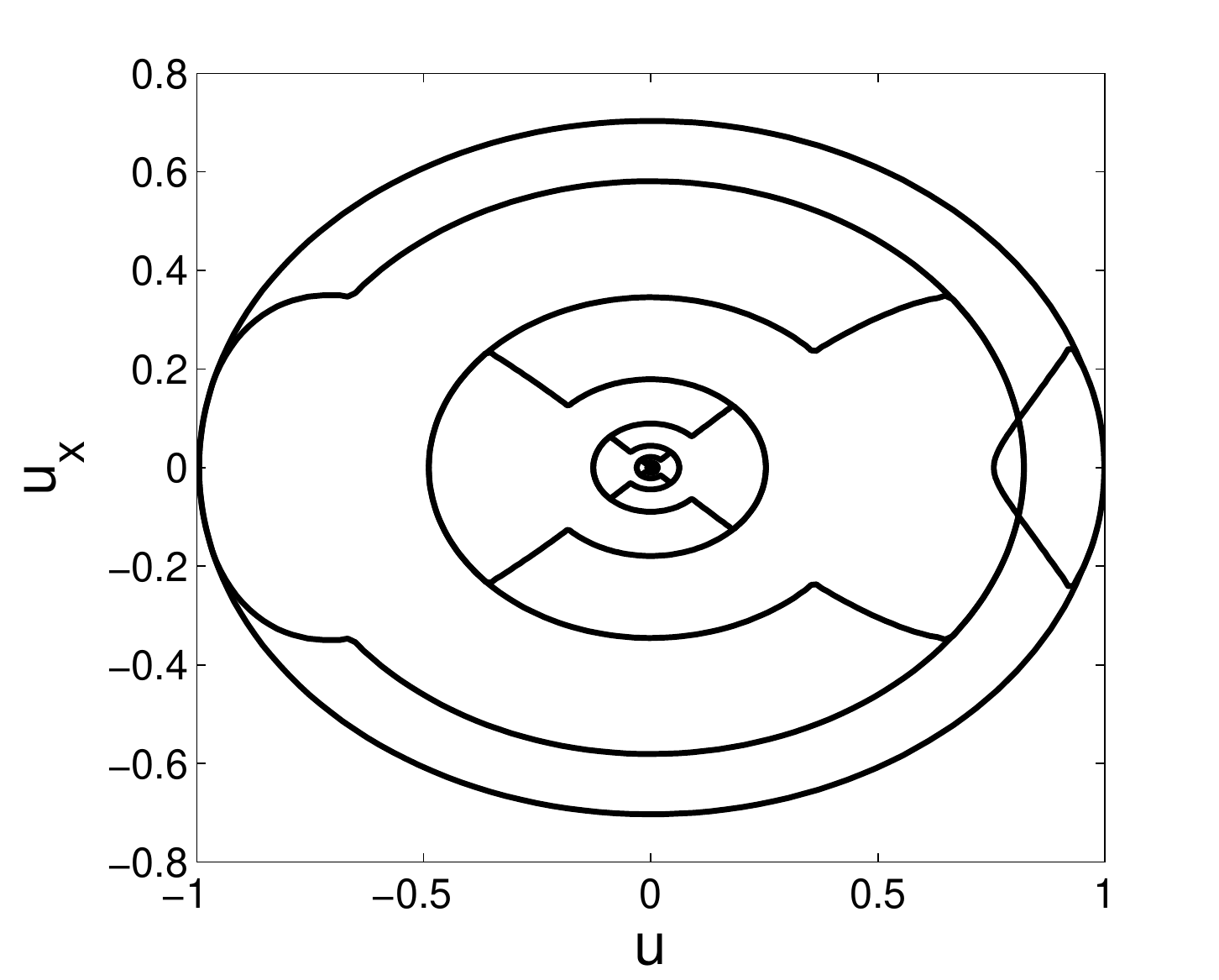}}
\caption{The plot of a gap soliton of (\ref{stat1}) in (a) the physical space, (b) the phase-space. The parameter values are explained in Section \ref{numeric}.}
\label{fig:gapsoliton}
\end{figure}

Let $U_S$ be the collection of intervals $U_S =[0,x_0)\cup (x_1,x_2) \cup (x_3,x_4) \ldots$. In the case of a gap soliton, $U_I = -U_S \cup
U_S$, and we have that the number of components of $U_I$ is infinite and the $x_i's$ are chosen so that the soliton travels from $(u_0,y_0)$
along the inner system to $(-u_0, -y_0)$. This is a key ingredient in the construction of the soliton, and will play a large role in establishing
instability. In \cite{komi06}, the inner system is linear, and the length of  the interval $(x_{2k} , x_{2k-1})$ can be determined as ${\pi}/{\sqrt{\w}}$. Here, we do not require that the inner system be linear, however we do require that the $x_i's$ be chosen so that if $i\geq 1$, the soliton travels from $(u_0,y_0)$ on the homoclinic orbit along the inner system to $(-u_0,-y_0)$, which is also on the homoclinic orbit.

In Figure \ref{fig:gapsoliton}, we plot an example of a gap soliton of the governing equation (\ref{gov1}) for parameter values that will be explained in Section \ref{numeric}. One can notice the main characteristic of gap solitons in the plot, which is the infinitely many zeros in the soliton profile.

\section{Instability Results}

To show instability of the standing waves, we will show that the matrix $M$ from above has a real positive eigenvalue. This is done by applying the main theorem of
\cite{ckrtj88}. In \cite{rmckrtjhs10}, systems like (\ref{gov1}) were considered with $U_I = (-L,L)$, for some real number $L$. One can show that the following quantities are well defined (see for example \cite{ckrtj88}, and the references therein):
\begin{eqnarray*}
P &=& \textrm{ the number of positive eigenvalues of } D_{+} \\
Q &=& \textrm{ the number of positive eigenvalues of } D_{-}.
\end{eqnarray*}
\noindent We then have the following:
\begin{theorem}[\cite{ckrtj88}] If $P-Q \neq 0, 1,$ there is a real positive eigenvalue of the operator $M$.  \label{th:ckrtj88}
\end{theorem}
From Sturm-Liouville theory, $P$ and $Q$ can be determined by considering solutions of $D_+ v = 0$ and $D_- v = 0$, respectively. In fact, they are the number of
zeros of the associated solution $v$. Notice that $D_- v = 0$ is actually satisfied by the standing wave itself, and that $D_+ v = 0$ is the equation of variations of the
standing wave equation. It follows that:
\begin{equation} \label{eq:pandq}
\begin{array}{lll}
&&Q = \textrm{ the number of zeros of the standing wave } u. \\
&&P = \textrm{ the number of zeros of a solution to the variational equation along $u$. }
\end{array}
\end{equation}

For gap solitons, it is not immediately clear how to apply Theorem \ref{th:ckrtj88} above as in this case, both, $P$ and  $Q\to\infty$. The idea presented in this paper is to build an approximation to a gap soliton using more and more intervals of $U_I$ for which the quantity $P-Q$ remains constant. To this end define $S_0 = [0, x_0) $ and $S_n = [0,x_0) \cup (x_1,x_2) \cup (x_3, x_4) \cup \ldots (x_{4n-1},x_{4_n})$, where $(x_i, x_{i+1}) \subseteq U_S$. Thus $S_n$ adds two more components for each $n$. Then we can define $U_n = -S_n \cup S_n$, and we let $f_n$ be a solution to the ODE
\begin{equation}\begin{array}{ccccc}
f_{xx} &=& (V-\w )f - f^3, & & x \in \R \setminus U_n, \\
f_{xx} &=& -\w f + \eta f^3, & & x \in U_n.
\end{array} \label{stat2}\end{equation}
Thus for example $f_0$ would be the solution to
\begin{equation}\begin{array}{ccccc}
f_{xx} &=& (V-\w )f - f^3, & & |x|\geq x_0, \\
f_{xx} &=& -\w f + \eta f^3, & & |x| < x_0,
\end{array} \label{stat3}\end{equation}
while $f_1$ would be a solution to
\begin{equation}\begin{array}{ccccc}
f_{xx} &=& (V-\w )f - f^3, && x \notin (-x_4,-x_3) \cup (-x_2,-x_1) \cup (-x_0, x_0) \cup (x_1,x_2) \cup (x_3,x_4)\\
f_{xx} &=& -\w f + \eta f^3, && x \in (-x_4,-x_3) \cup (-x_2,-x_1) \cup (-x_0, x_0) \cup (x_1,x_2) \cup (x_3,x_4).
\end{array} \label{stat4}\end{equation}
A gap soliton then can be realized as the limit of successive $f_n$'s (in a variety of norms, but in particular in the $L^2$ and $H^1$ norms). In Figure \ref{fig:successivegapsolitons} we present a plot of $f_n$, $n=0,1,2$, approximating the gap soliton in Figure \ref{fig:gapsoliton}.

\begin{figure}
\center
\subfigure[]{\includegraphics[scale=.4]{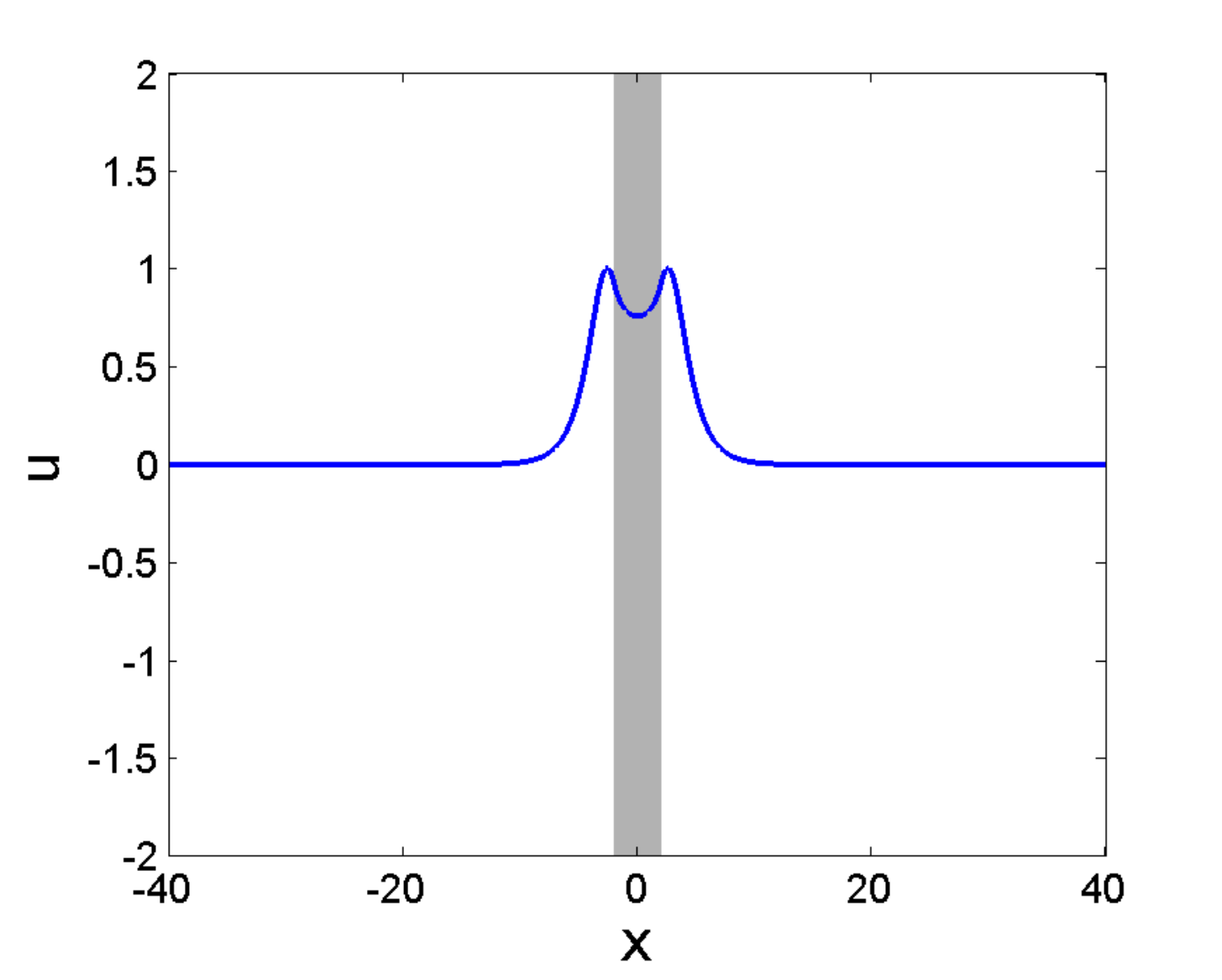}}
\subfigure[]{\includegraphics[scale=.4]{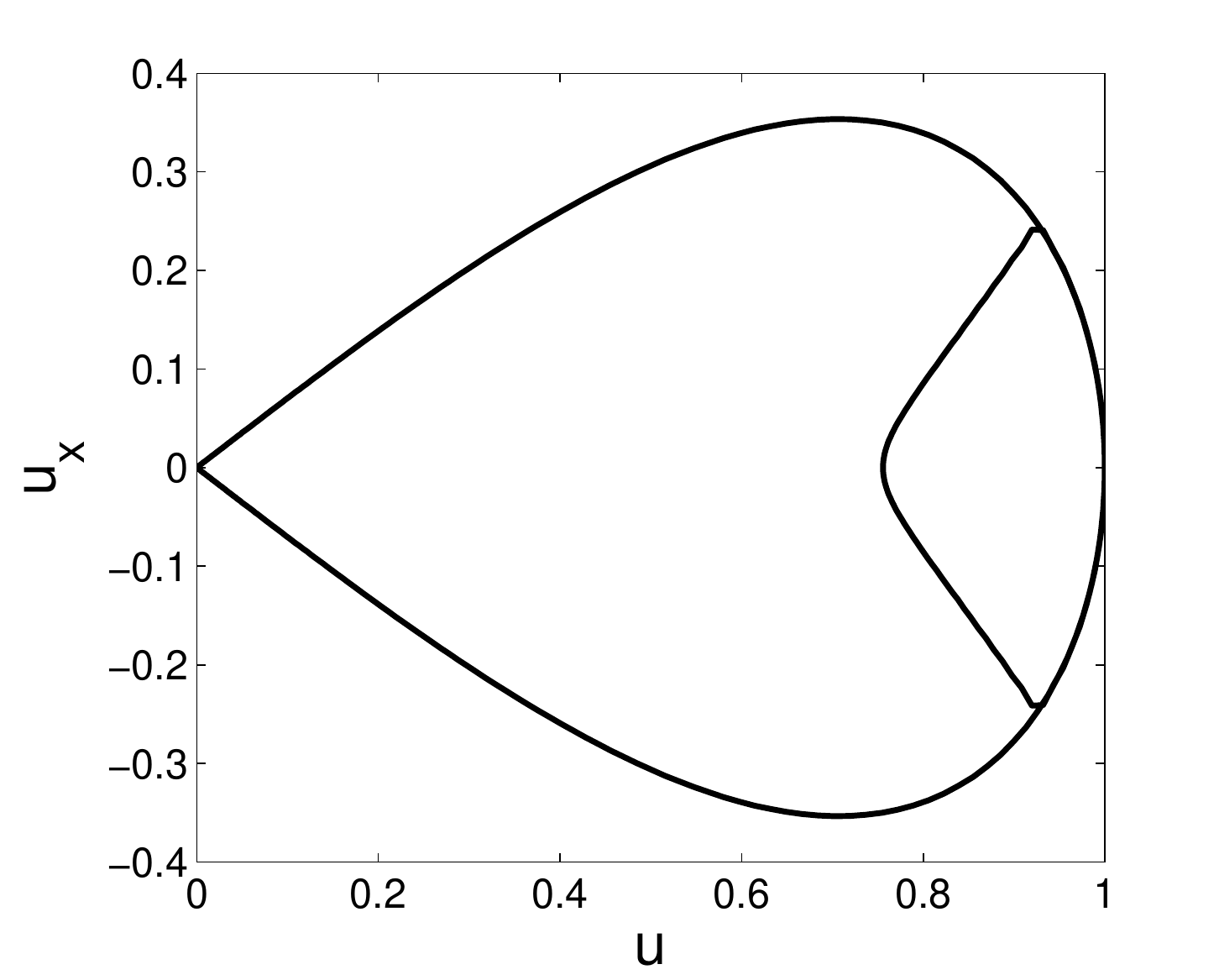}}
\subfigure[]{\includegraphics[scale=.4]{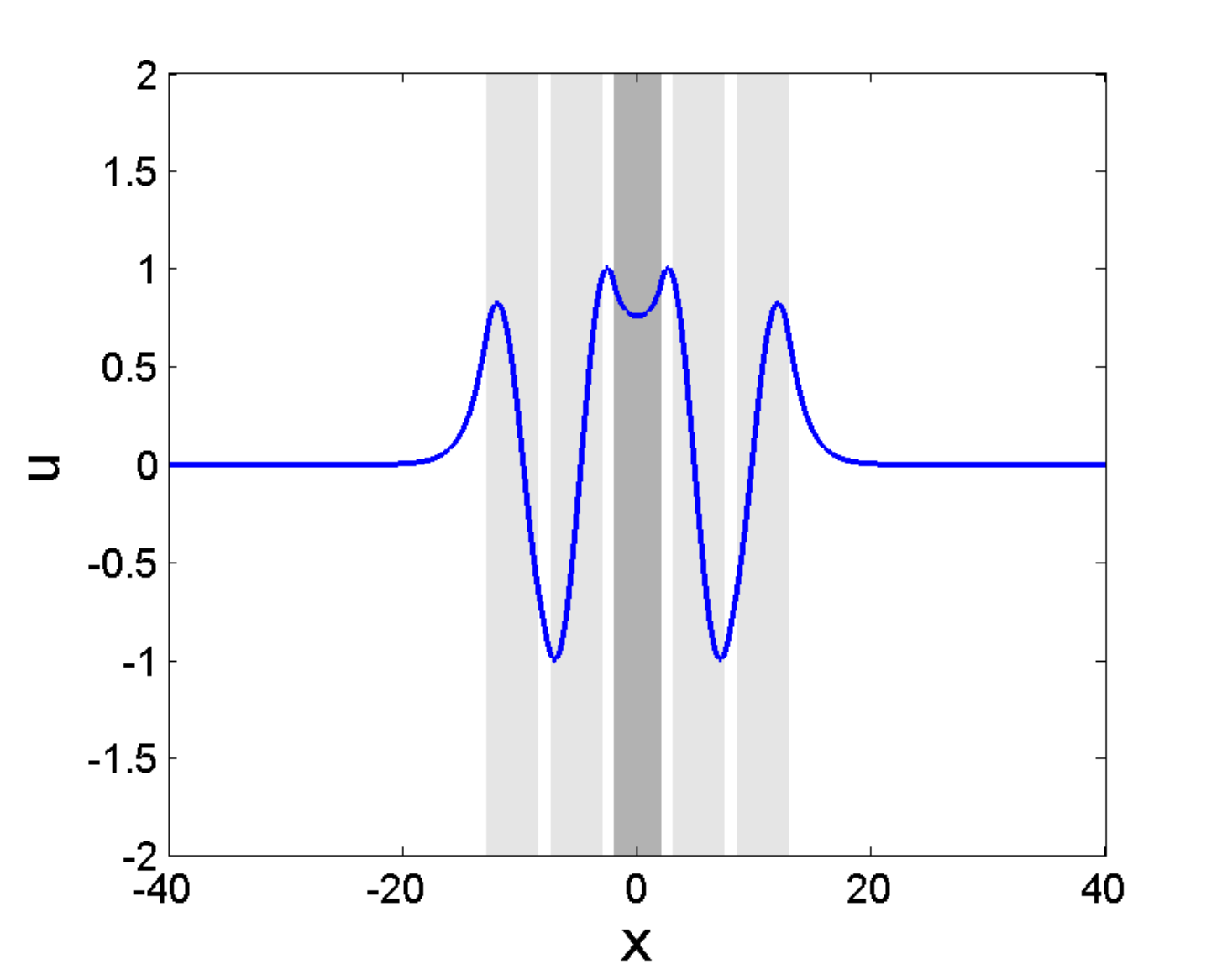}}
\subfigure[]{\includegraphics[scale=.4]{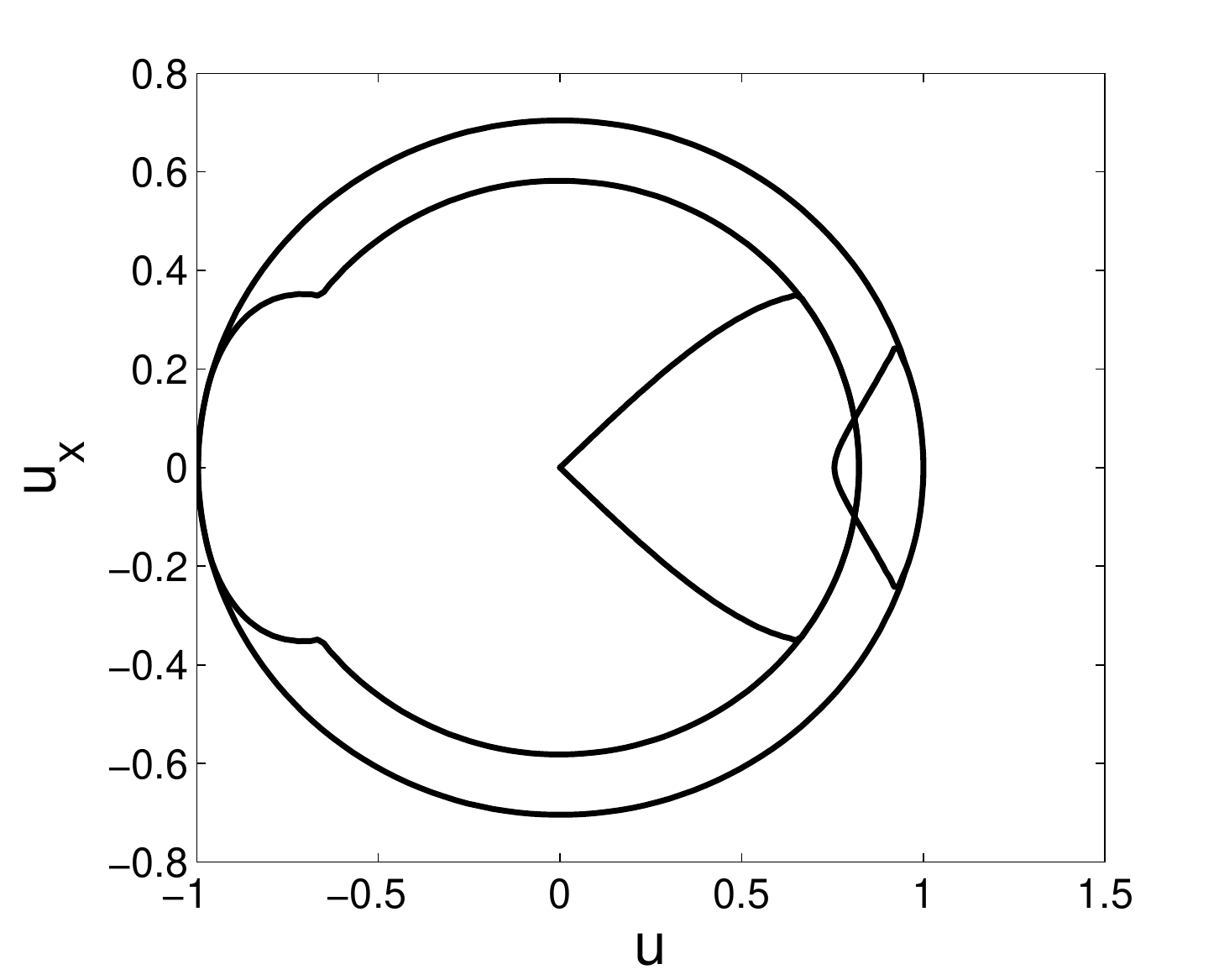}}
\subfigure[]{\includegraphics[scale=.4]{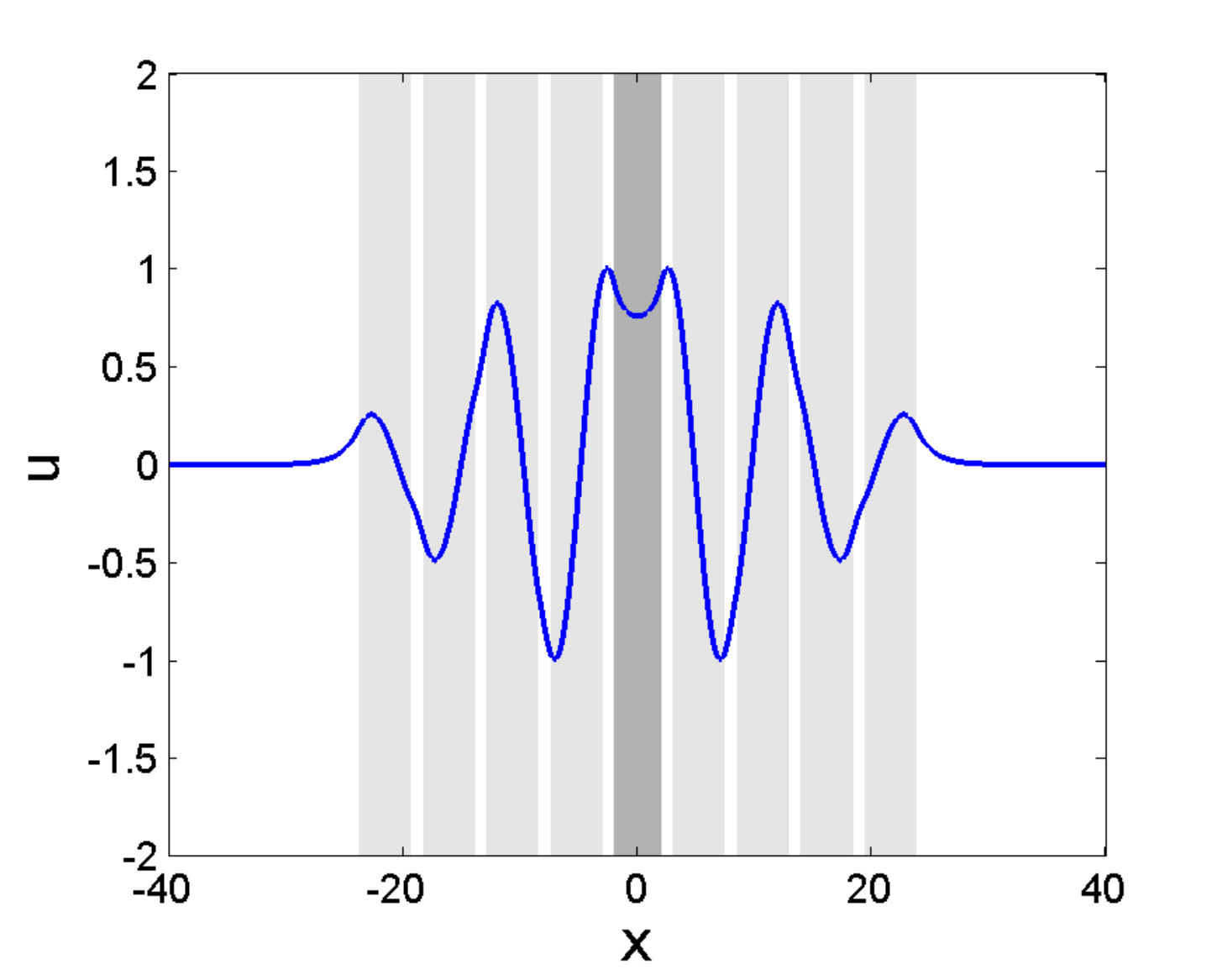}}
\subfigure[]{\includegraphics[scale=.4]{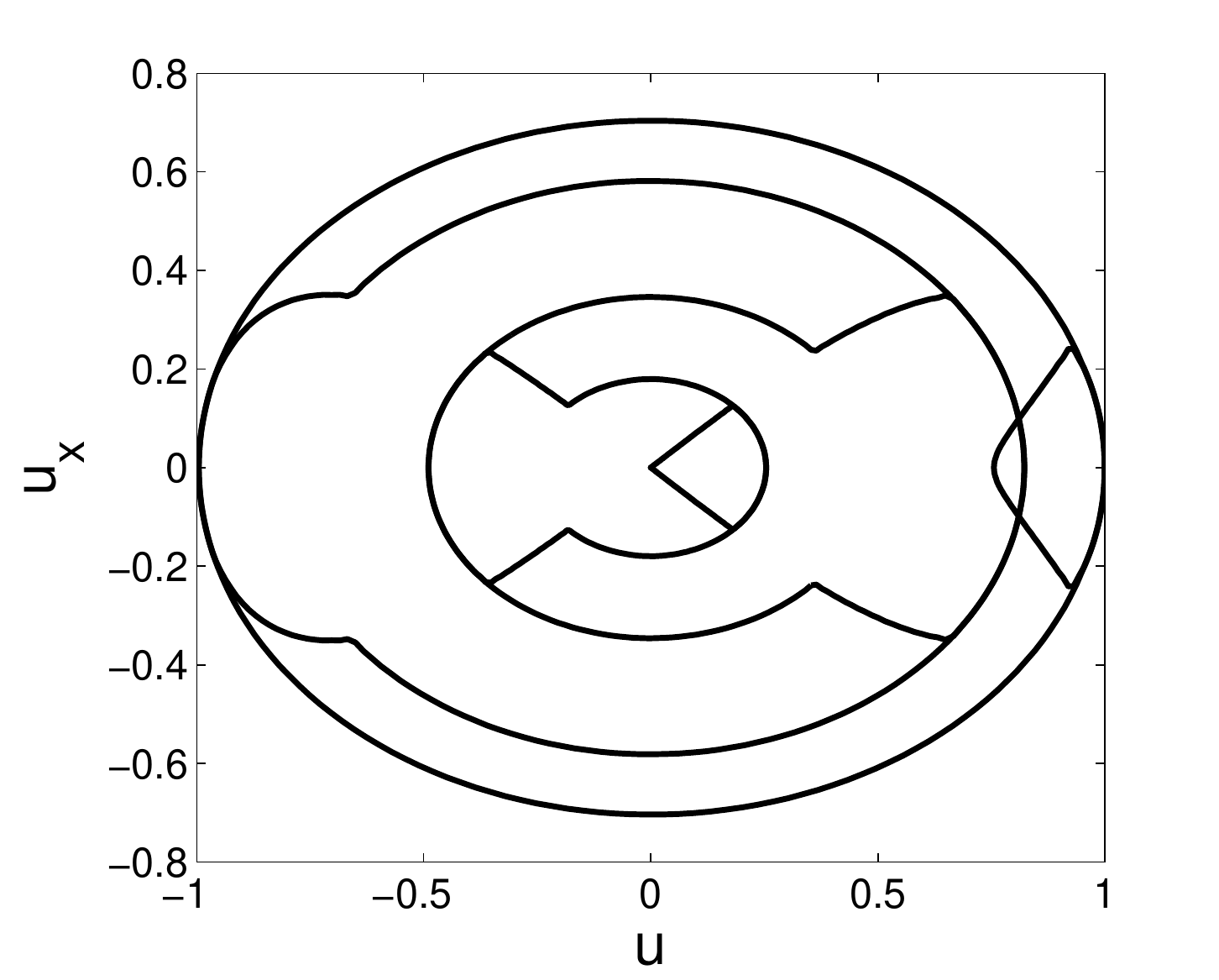}}
\caption{Successive approximations to a gap soliton in Figure \ref{fig:gapsoliton} in the physical space (a,c,e) and in the phase-space (b,d,f). The first, second and third row is respectively $f_0$, $f_1$, and $f_2$.}
\label{fig:successivegapsolitons}
\end{figure}

We have the following theorem
\begin{theorem}\label{th:main}
The quantity $P-Q$ is the same for all $f_i$ described above. Thus if $f_0$ is unstable then so is $f_n$ for all $n$. Further, if $f_0$ is unstable, then so is $f$, the gap soliton, corresponding to the limit.
\end{theorem}

The key idea is to use the interpretation of $P$ and $Q$ given in (\ref{eq:pandq}) as the number of zeros of the solution $f$ and the number of zeros of the solution to the variational equation along $f$, for the partial solution defined on $(x_i, x_{i+4})$, to the ODE below:
\begin{equation}\begin{array}{ccccc} 
f_{xx} &=& (V-\w )f - f^3, & & x \in (x_{i+1},x_{i+2}) \cup (x_{i+3}, x_{i+4}) \\
f_{xx} &=& -\w f + \eta f^3, & & x \in (x_i,x_{i+1}) \cup (x_{i+2},x_{i+3}).
\end{array} \label{stat5}
\end{equation}

\begin{figure}
\center \includegraphics[scale=0.8]{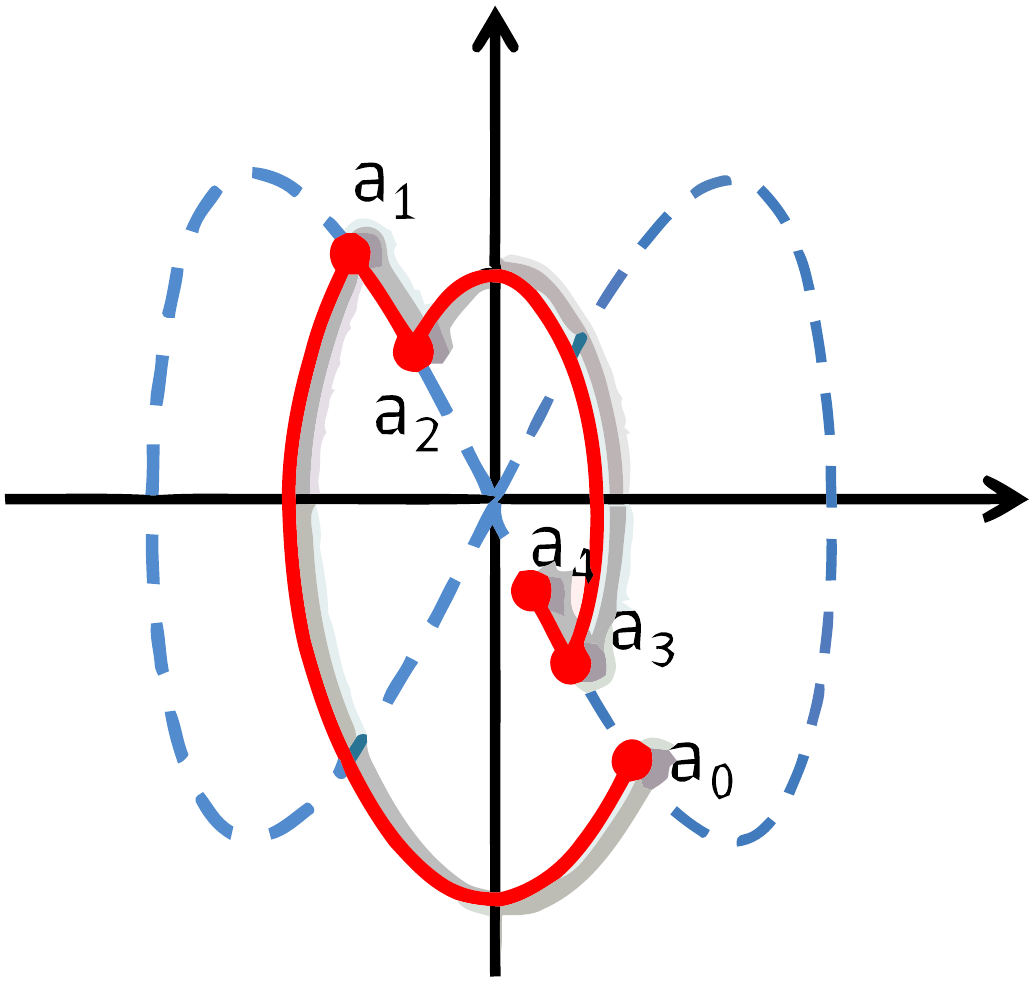}
\caption{A sketch of a phase portrait of the partial solution to equation (\ref{stat5}). The points $a_i$ correspond to the points $(f(x_{i-1}),f_x(x_{i-1}))$ in the phase plane.}
\label{fig:gapdetour}
\end{figure}

The number $Q$ is straight forward to calculate. We make the geometric observation as in \cite{rmckrtjhs10} that $P$, the number of zeros of a solution to the equation of variations along $f$, can be found by determining the number of times that a vector must pass through the vertical as the base point ranges over the
entire orbit. It turns out that for the solution of (\ref{stat5}) defined above, the rotation of a vector by the equation of variations is the same (mod $2 \pi)$ as if the base point had traveled along only the outer homoclinic orbit.
\begin{example}\label{ex:lin}
To better illustrate this last point, we first consider the case when both the inner system and the outer systems are linear. That is we have the following systems of linear, constant coefficient equations
\begin{eqnarray}
\begin{pmatrix} u \\ y \end{pmatrix}_x & = &
\begin{pmatrix} 0 & 1 \\ (V-\w) & 0 \end{pmatrix} \begin{pmatrix} u \\ y \end{pmatrix},\, \textrm{ when } x \in (x_{i+1}, x_{i+2}) \cup (x_{i+3},x_{i+4}) \label{eq:outersim} \\
& = & \begin{pmatrix} 0 & 1 \\ -\w & 0  \end{pmatrix} \begin{pmatrix} u \\ y \end{pmatrix},\, \textrm{ when } x \in (x_i, x_{i+1}) \cup
(x_{i+2},x_{i+3}) \label{eq:innersim}
\end{eqnarray}
The solution to the above equation can be written explicitly. Further, because we are in the linear case, we have that the equation of variations along a solution is the same as the equation itself (\ref{eq:outersim} \ref{eq:innersim}).

Being led by the geometry of the phase plane, we let $\Phi_1(a,b)$ denote a fundamental solution matrix to the equation of variations of the outer system of equations (\ref{eq:outersim}) along a solution to (\ref{eq:outersim}) which travels from point $a$ to point $b$ in the phase plane. That is let $(u(x),y(x))$ be a solution to (\ref{eq:outersim}), considered on the interval $(x_j,x_k)$. Then set $a := (u(x_j), y(x_j))$ and $b :=(u(x_k),y(x_k))$, and define $\Phi_1(a,b)$ to be a fundamental solution matrix of the equation of variations to the outer system, along the path $(u(x), y(x))$ with $x \in (x_j, x_k)$. 

Similarly, let $\Phi_2(a,b)$ be a fundamental solution matrix to the equation of variations
of the inner system (\ref{eq:innersim}), along a solution to (\ref{eq:innersim}) evolving from point $a$ to point $b$. We denote by $a_0, a_1, a_2, a_3$, the points in the phase plane of (\ref{eq:outersim}, \ref{eq:innersim}) where the
solutions switches between the two systems, and $a_4$ the point where we stop evolving (see Figure \ref{fig:gapdetour}), and we let $\begin{pmatrix} \zeta_0 \\ \xi_0 \end{pmatrix}$ be a pair of initial conditions in the tangent plane to $\R^2$ at the point
$a_0$. We have that a solution to the equation of variations along the orbit from $a_0$ to $a_1$ to $a_2$ to $a_3$ to $a_4$ can be described as
$$\Phi_1(a_3,a_4) \Phi_2(a_2,a_3)\Phi_1(a_1,a_2)\Phi_2(a_0,a_1) \begin{pmatrix} \zeta_0 \\ \xi_0 \end{pmatrix}.$$
It turns out that modulo $2 \pi$,
\begin{equation}\label{eq:rotid}
\Phi_1(a_3,a_4) \Phi_2(a_2,a_3)\Phi_1(a_1,a_2)\Phi_2(a_0,a_1) \begin{pmatrix} \zeta_0 \\ \xi_0 \end{pmatrix} = \Phi_1(a_0,a_4) \begin{pmatrix} \zeta_0 \\ \xi_0 \end{pmatrix}.
\end{equation}
\noindent The equality in equation (\ref{eq:rotid}) can be verified by solving the appropriate systems. Another way to see the effect is to
consider the following. As the base point evolves under equation (\ref{eq:outersim} \ref{eq:innersim}) from $a_i$ to $a_{i+1}$, we can consider the aggregate effect of a $\Phi_j(a_i, a_{i+1})$ on a
tangent vector $\begin{pmatrix} \zeta_0 \\ \xi_0 \end{pmatrix}$, as a linear map from $\R^2 \to \R^2$, by simply determining where a tangent
vector to $a_i$ gets sent to, when the base point is at $a_{i+1}$. That is we are considering $\Phi_j(a_i,a_{i+1})$ as a map between the tangent plane of $\R^2$ at the point $a_i$ to the tangent plane of $\R^2$ at the point $a_{i+1}$. This will give us the total rotation of a tangent vector modulo $2 \pi$ as we
travel from point $a_i$ to point $a_{i+1}$ along the orbit. The key observation is to realize that for $\Phi_2(a_0,a_1)$ and $\Phi_2(a_2,a_3)$, this will be negative the identity $- \id$. That is, viewing $\Phi_2(a_j,a_{j+1}) \quad j = 0,2$ as a map between tangent spaces of $\R^2$,  
$\Phi_2(a_j,a_{j+1}):T_{a_j}\R \to T_{a_{j+1}}\R \quad j = 0,2 $, we have $\Phi_2(a_j,a_{j+1}) = -\id$. Moreover, by considering $\Phi_2(a_j,a_{j+1})$ in this way, we are just measuring the effect of rotation by $\Phi_2(a_j,a_{j+1})$ on an initial tangent vector modulo $2 \pi$, and we have that
\begin{equation} \label{eq:negid}\begin{array}{lll}
&&\Phi_1(a_3,a_4) \Phi_2(a_2,a_3)\Phi_1(a_1,a_2)\Phi_2(a_0,a_1) \begin{pmatrix} \zeta_0 \\ \xi_0 \end{pmatrix}\\
&&=  (-\id)^2 \Phi_1(a_3,a_4) \Phi_1(a_1,a_2) \begin{pmatrix} \zeta_0 \\ \xi_0 \end{pmatrix} \\
&&=  \Phi_1(a_0, a_4) \begin{pmatrix} \zeta_0 \\ \xi_0 \end{pmatrix},\end{array}
\end{equation}
where the last equality follows from the facts that $a_0 = -a_1$, $a_2=-a_3$, the outer system of equations (\ref{eq:outersim}) is symmetric about the origin, and the group property of variational flows.
\end{example}
 We are now ready to state the main lemma used in the proof of theorem \ref{th:main}.

\begin{lemma}\label{lem:main} Redefine $\Phi_1(a,b)$ and $\Phi_2(a,b)$ as in the above example, but instead of using the linear ODE, let them be the fundamental solution matrices to the equations of variations along solutions to the inner and outer systems given in the nonlinear equation (\ref{stat5}):
\begin{equation*}\begin{array}{ccccc} 
f_{xx} &=& (V-\w )f - f^3, & & x \in (x_{i+1},x_{i+2}) \cup (x_{i+3}, x_{i+4}), \\
f_{xx} &=& -\w f + \eta f^3, & & x \in (x_i,x_{i+1}) \cup (x_{i+2},x_{i+3}).
\end{array} \end{equation*}
Likewise, let $a_j$ be defined analogously for the points in the phase plane of the nonlinear equation where the orbit switches between the inner and outer systems.  Also, let $\displaystyle \begin{pmatrix} \zeta_0 \\ \xi_0 \end{pmatrix}$ be an initial condition to the equation of variations along a solution to (\ref{stat5}) in the tangent plane to $\R^2$ at $a_0$. Then we have the following:
\begin{equation}\label{eq:rotidnl}
\Phi_1(a_3,a_4) \Phi_2(a_2,a_3)\Phi_1(a_1,a_2)\Phi_2(a_0,a_1) \begin{pmatrix} \zeta_0 \\ \xi_0 \end{pmatrix} = \Phi_1(a_0,a_4) \begin{pmatrix} \zeta_0 \\ \xi_0 \end{pmatrix}.
\end{equation}
\end{lemma}

\begin{proof}
The exact same reasoning can be used to prove Lemma \ref{lem:main} (the nonlinear case), as was used in the example (the linear case). The only difference is that in order to determine the aggregate effect of the inner system on an initial tangent vector some more care must be taken with the matrices $\Phi_2(a_i, a_{i+1})$. Write the equation of variations to the outer system as
\begin{equation}\label{eq:outrsys}
\begin{pmatrix} \zeta \\ \xi \end{pmatrix} = \begin{pmatrix} 0 & 1 \\ -3u_1^2 + V- \w & 0 \end{pmatrix} \begin{pmatrix} \zeta \\ \xi \end{pmatrix}, \,\textrm{ when } x \in (x_{i+1},x_{i+2}) \cup (x_{i+3},x_{i+4}),
\end{equation}
where $u_1(x)$ is the equation satisfying the outer system with $\lim_{x\to\pm \infty} u_1(x) = \lim_{x\to\pm \infty} u_1'(x) = 0$. Write the equation of variations of the inner system as
\begin{equation}\label{eq:innersys}
\begin{pmatrix} \zeta \\ \xi \end{pmatrix} = \begin{pmatrix} 0 & 1 \\ 3\eta u_2^2 - \w & 0 \end{pmatrix} \begin{pmatrix} \zeta \\ \xi \end{pmatrix}, \,\textrm{ when } x \in (x_i,x_{i+1}) \cup (x_{i+2},x_{i+3}),
\end{equation}
where $u_2^2$ satisfies the appropriate conditions for the orbit. Now here is where the appropriate choices of the $x_i$'s must come into play.
In the linear case, the $x_i$'s were chosen so that the length of an interval in $U_I$ was $\frac{\pi}{\sqrt{\w}}$. Here we choose the $x_i$'s in
$U_I$ so that the length of an interval is such that we will return not only to the homoclinic orbit, but also if we leave the homoclinic orbit at the point $(u_0,y_0)$, we will return to the homoclinic orbit at the point
$(-u_0,-y_0)$. This allows us to determine the effect of the rotation (modulo $2\pi$) by the flow associated to the equation of variations along the partial orbit 
$(u_2(x),y_2(x)) $. In fact, we claim that the exact same is true as in the linear case. If $B$ is the linear map from the tangent space at $a_0$ and at $a_2$
to the tangent spaces at $a_1, a_3$ respectively, then $B = -\id$. To see this we will write out $B$ in a suitable basis $\vec{v}_1, \vec{v}_2$ of
the tangent space at $a_0$. One obvious choice of a basis vector is the tangent vector to the inner system. However given equation
(\ref{eq:innersys}), and the fact that along an orbit $(u_0,y_0) \to (-u_0,-y_0)$, this means that if $\vec{v}_1$ is the vector tangent to the inner orbit at $a_0$ (or $a_2$), then under $B$ $\vec{v}_1 \to -\vec{v}_1$. This means that $B$ has the form:
\begin{equation}
B = \begin{pmatrix} -1 & b_{1,2} \\ 0 & b_{2,2} \end{pmatrix},
\end{equation}
where $b_{i,j}$ are the coefficients of the linear combination of $\vec{v_1}$ and a suitably chosen $\vec{v}_2$. Now we appeal to two facts about the matrix $B$ which are evident from it's definition. The first is that $B$ must be orientation preserving. This is an elementary consequence due of the fact that it is the matrix of a flow (see for example \cite{perko01}). This means that $b_{2,2}$ must be negative. The second fact is that since $B$ corresponds to the matrix of the equation of variations traveling half way along the periodic orbit given by $(u_2(x),y_2(x))$ (because we chose our $x_i$'s so it would be that way), we must have that $B^2 = \id$. But this means that $b_{1,2} =0$ and $b_{2,2} = -1$ and the matrix $B$ itself $B = -\id$. Now we simply repeat the computation done in equation (\ref{eq:negid}) and the proof of Lemma \ref{lem:main} is complete.
\end{proof}

We are now ready to complete the proof of theorem \ref{th:main} .

\begin{pot}
Recall that $f_n$ as constructed is the solution to the ODE (\ref{stat2}). We let $P_n$ and $Q_n$ denote the count for $f_n$ of $P$ and $Q$ respectively. Lemma \ref{lem:main} shows that $P_{n-1} = P_n+2$ and it is clear that $Q_{n-1} = Q_{n}$, and so the quantity $P_n-Q_n$ is the same for all $f_n$, and in particular is equal to $P-Q$ for $f_0$. This completes the first part of the proof of theorem \ref{th:main}.

In order to determine the instability of the limit soliton we must proceed topologically using the methods developed in the proof of the main theorem of \cite{ckrtj88}.  

We have already discussed that in $H^1$, $f_n \to f$ a solution to
\begin{equation}\begin{array}{ccccc}
f_{xx} &=& (V-\w )f - f^3 & & x \in \R \setminus U_I, \\
f_{xx} &=& -\w f + \eta f^3 & & x \in U_I.
\end{array} \label{top2}\end{equation}

Following \cite{ckrtj88} we can associate to each solution $f_n$ a curve $\gamma_n(x)$, and to $f$ a curve $\gamma(x)$ in $\Lambda(2)$ the 
space of Lagrangian planes in $\R^4$. 

This is done as follows. Let $\Phi^n_{L_+}(x)$ denote the evolution operator of the ODE corresponding to 
the equation of variations of the ODE (\ref{stat2}) along the solution $f_n$. Likewise, let $\Phi_{L_+}(x)$ denote the evolution operator of the ODE 
corresponding to the equation of variation of the ODE (\ref{top2}) along the solution $\displaystyle f = \lim_{n \to \infty} f_n$.  Thus if 
$\displaystyle \begin{pmatrix} v_0 \\ w_0 \end{pmatrix}$ is a pair of initial conditions at $x = 0$, then for any $x \in \R$ we have that the 
evolution of $\displaystyle \begin{pmatrix} v_0 \\ w_0 \end{pmatrix}$ under the equation of variations along $f$, $f_n$ respectively will be 
given by $\displaystyle \begin{pmatrix}  \Phi_{L_+}(x) \cdot v_0 \\ \Phi_{L_+}(x)\cdot w_0 \end{pmatrix}$, respectively $\displaystyle \begin{pmatrix} 
\Phi_{L_+}^n(x) \cdot v_0 \\\Phi_{L_+}^n(x) \cdot w_0 \end{pmatrix}$.

We remark that the initial conditions $\displaystyle \begin{pmatrix} v_0 \\ w_0 \end{pmatrix} $ will be the same for each $f_n$ as well as for $f$. 

Again appealing to \cite{ckrtj88}, we can explicitly write the curves $\gamma_n(x)$ and $\gamma(x)$ in the space of Lagrangian planes $\Lambda(2) \approx U(2)/O(n)$. This is given by 
\begin{equation}
\gamma_n(x) = \begin{pmatrix} e^{i \theta_{1,n}(x)} & 0 \\ 0 & e^{i \theta_{2,n}(x)} \end{pmatrix},
\end{equation}
\noindent where
\begin{equation}
\theta_{1,n} = 2 \arctan(\frac{\Phi_{L_+}^n(x) \cdot w_0}{\Phi_{L_+}^n(x) \cdot v_0}) \textrm{ and, } \theta_{2,n} = -2 \arctan(\frac{f_n'(x)}{f_n(x)}),
\end{equation}
\noindent and
\begin{equation}
\gamma(x) = \begin{pmatrix} e^{i \theta_1(x)} & 0 \\ 0 & e^{i \theta_2(x)} \end{pmatrix},
\end{equation}
\noindent where
\begin{equation}
\theta_1 = 2 \arctan(\frac{\Phi_{L_+}(x) \cdot w_0}{\Phi_{L_+}(x) \cdot v_0}) \textrm{ and, } \theta_2 = -2 \arctan(\frac{f'(x)}{f(x)}).
\end{equation}
Now we observe that the curves $\gamma_n(x)$ and $\gamma(x)$ actually lie on a torus contained in $\Lambda(2)$. 

It was established in \cite{ckrtj88} that because $f_n$ and $f$ are solutions corresponding to homoclinic orbits in the phase plane of equations (\ref{stat2}),and (\ref{top2}), the curves $\gamma_n(x)$, and $\gamma(x)$ have well defined end points. Let $\mu_{-,n}$, $\mu_{+,n}$ be the endpoints in $\Lambda(2)$ of $\gamma_n(x)$. That is, let 
$$
\lim_{x \to -\infty} \gamma_n(x) = \mu_{-,n} \textrm{ and, } lim_{x \to \infty} \gamma_n(x) = \mu_{+,n},
$$
\noindent and set 
$$
\lim_{x \to -\infty} \gamma(x) = \mu_{-} \textrm{ and, } lim_{x \to \infty} \gamma(x) = \mu_{+}.
$$
Further because $f_n \to f$ and lemma \ref{lem:main}, we have that $\mu_{-,n} = \mu_{-}$, and $\mu_{+,n} = \mu_{+}$ for all $n$. In the previously introduced coordinates on the torus in $\Lambda(2)$ this means that the limits of $\theta_{1,n}$, $\theta_{2,n}$ are equal to the limits of $\theta_1(x)$ and $\theta_2(x)$ as $\x \to \pm \infty$. Moreover, it is easy to calculate explicitly that
$$\theta_{1}(x) \to 2 \arctan(\sqrt{V-\w}) := \theta_- \textrm{ and,} \quad \theta_2(x) \to - \theta_-$$
\noindent as $x \to -\infty$. 

\begin{figure}[tbhp!]
\center{
{\includegraphics[scale=0.7]{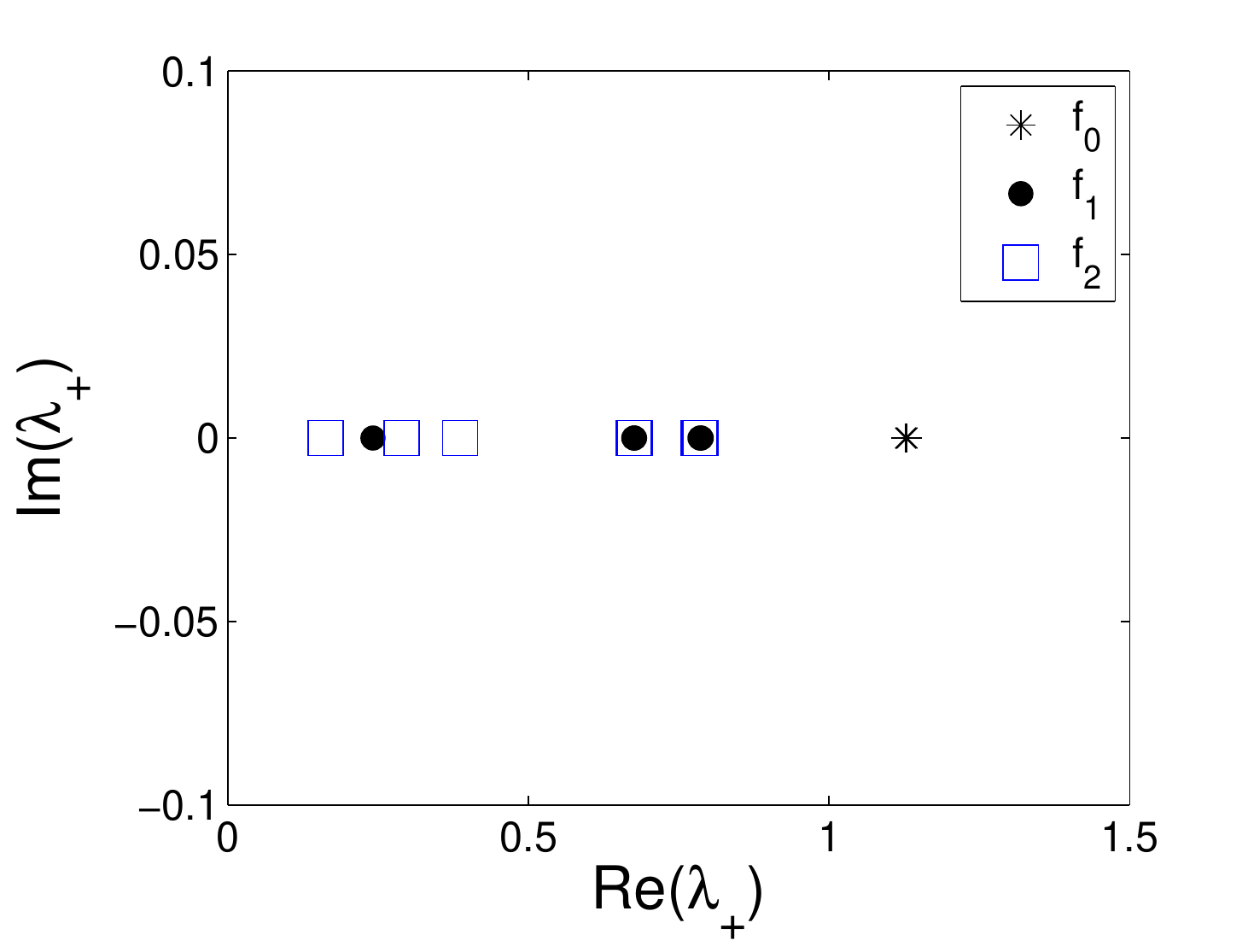}}
}
\caption{The positive eigenvalues of the operator $D_+$, i.e.\ $\lambda_+$. One symbol corresponds to two different, but very close eigenvalues.}
\label{fig:lambda_p}
\end{figure}

Still following the outline laid out in \cite{ckrtj88}, we denote by $\tilde{}$ the lift of the point (or curve) in the torus embedded in $\Lambda(2)$ to its corresponding point in the universal cover of the torus, $\R^2$. We will parametrize the universal covering of the torus in the obvious way. 
Without loss of generality, all of the $\mu_{-,n}$'s and $\mu_{-}$ can be lifted to the same point $\tilde{\mu}_- = (\theta_-, -\theta_-)$. It was shown in \cite{ckrtj88} that for each $n$, $\mu_{+,n}$ lifts to the point $\tilde{\mu}_{+,n} = ( \pm \theta_-, \theta_- + (P-Q) 2 \pi$. Thus lemma \ref{lem:main} implies that each $\mu_{+,n}$ lifts to the same point $\tilde{\mu}_{+,0} = (\pm \theta_-, \theta_-  + 2 \pi k )$, 

Next we observe that as $f_n \to f$ pointwise, $\gamma_n(x) \to \gamma(x)$ in the torus inside $\Lambda(2)$ pointwise, and the compactness of the torus and of $\Lambda(2)$, means that the end point $\mu_+$ must lift to the same point in the cover as $\mu_{+,0}$. Thus we have that $\tilde{\mu}_+ = (\pm \theta_- , \theta_- + 2 \pi k )$. 

Finally, it was shown in \cite{ckrtj88}, that if $|k| \neq 0, 1$, then the corresponding soliton underlying the curve $\gamma$ is unstable. This completes the proof of theorem \ref{th:main}
\end{pot}

\begin{rmk} The proof of theorem 2 may also be couched in the language of fixed end point homotopy classes. There are several ways to define such classes, see for example \cite{robsal93} or \cite{abbond01}, and the references therein. In this context theorem \ref{th:main} establishes that the fixed end-point homotopy class of the curve $\gamma$ is the same as those for $\gamma_n(x)$. An immediate consequence of this observation is that in $\Lambda(2)$, it is possible to deform the curves $\gamma$, and $\gamma_n$ all to the curve $\gamma_0$, in a continuous way. 
\end{rmk}

\begin{rmk}
\label{remark}
One can also consider so-called {\em surface} gap solitons, and obtain exactly the same results as for theorem \ref{th:main}. Mathematically, a surface gap soliton is the evolution of the solution to equation (\ref{top2}) but with the chosen intervals $U_I$ replaced by $U_S$, defined earlier. In this case, we consider a sequence of functions $f_n$ which are solutions to equation (\ref{stat2}), but with $U_n$ replaced by $S_n$. Then the functions $f_n \to f$, a solution to (\ref{top2}) with the appropriate replacements.  Lemma \ref{lem:main} holds, as well as theorem \ref{th:main}, and the techniques used in each will be identical. Thus if we start with an unstable solution, then the surface gap soliton that we obtain in  the limit will also be unstable. (See below for a further discussion of surface gap solitons). 
\end{rmk}

\begin{figure}[tbhp!]
\center{
\subfigure[]{\includegraphics[scale=0.6]{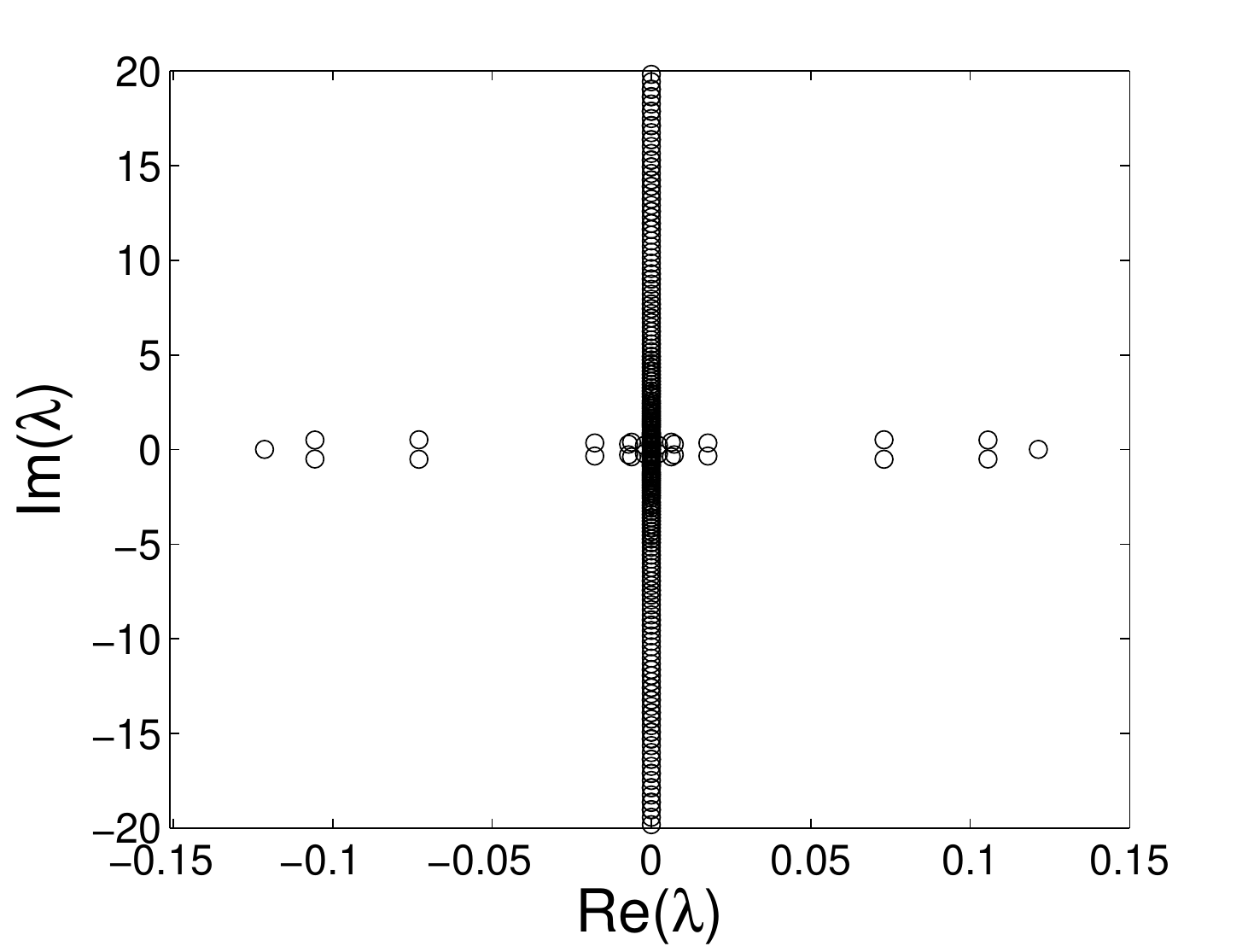}}
\subfigure[]{\includegraphics[scale=0.6]{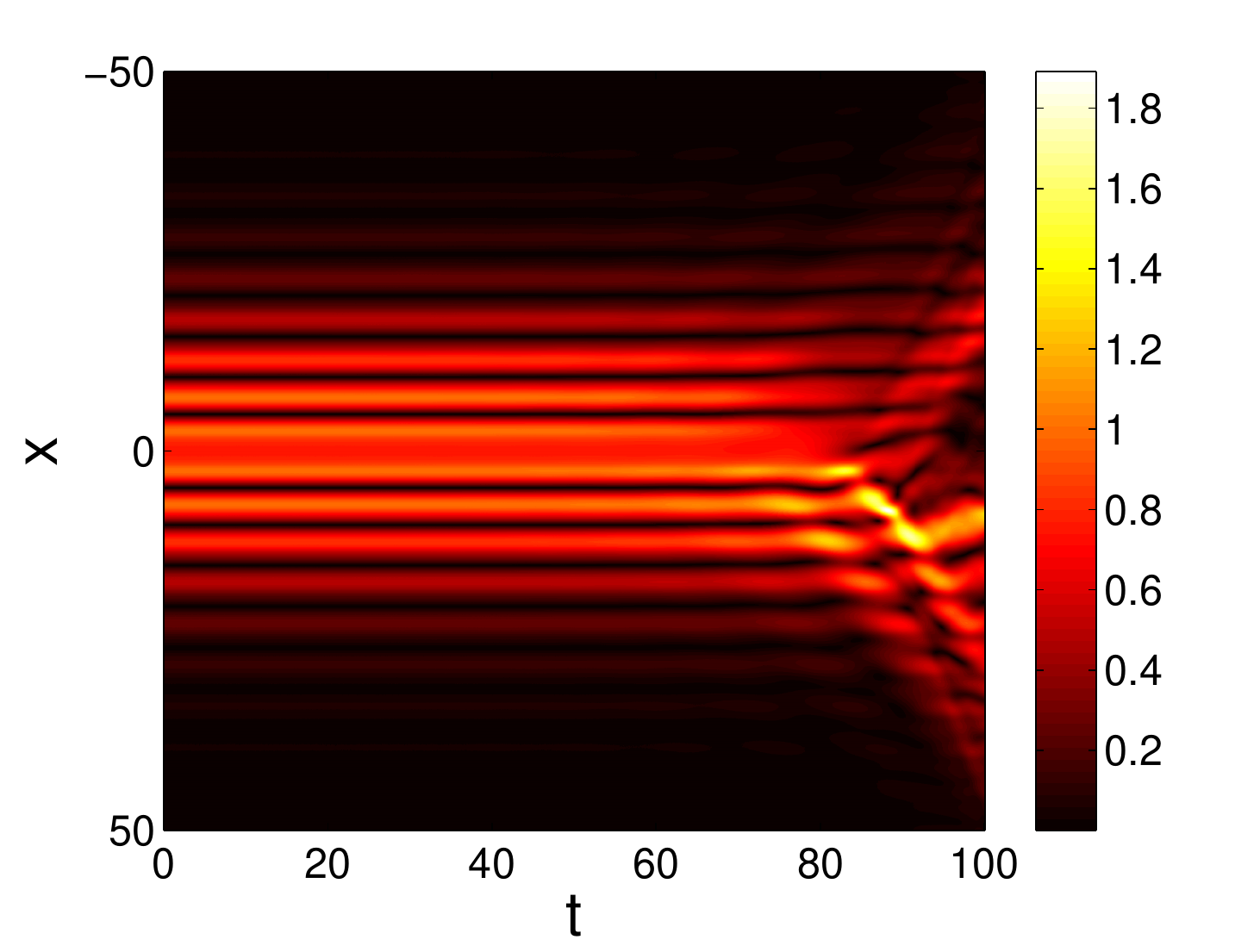}}
}
\caption{The eigenvalue structure in the complex plane (a) and the typical time evolution (b) of the gap soliton in Figure \ref{fig:gapsoliton}. Shown in (b) is the top view of $|\psi(x,t)|$ in the $(t,x)$-plane.}
\label{fig:lambda}
\end{figure}

\section{Numerical solutions and Discussion}
\label{numeric}

We have solved the time independent equation (\ref{stat1}) numerically, where we have used a spectral difference method to approximate the Laplacian $u_{xx}$. Once a solution is obtained, the corresponding eigenvalue problem (\ref{eq:linear}) is solved using a MATLAB routine. The time dependent equation (\ref{gov1}) is integrated numerically using a fourth-order Runge-Kutta method. Throughout the paper, we consider the parameter values
\[
V=1,\,\omega=0.5.
\]

\begin{figure}[tbhp!]
\center{
\subfigure[]{\includegraphics[scale=0.4]{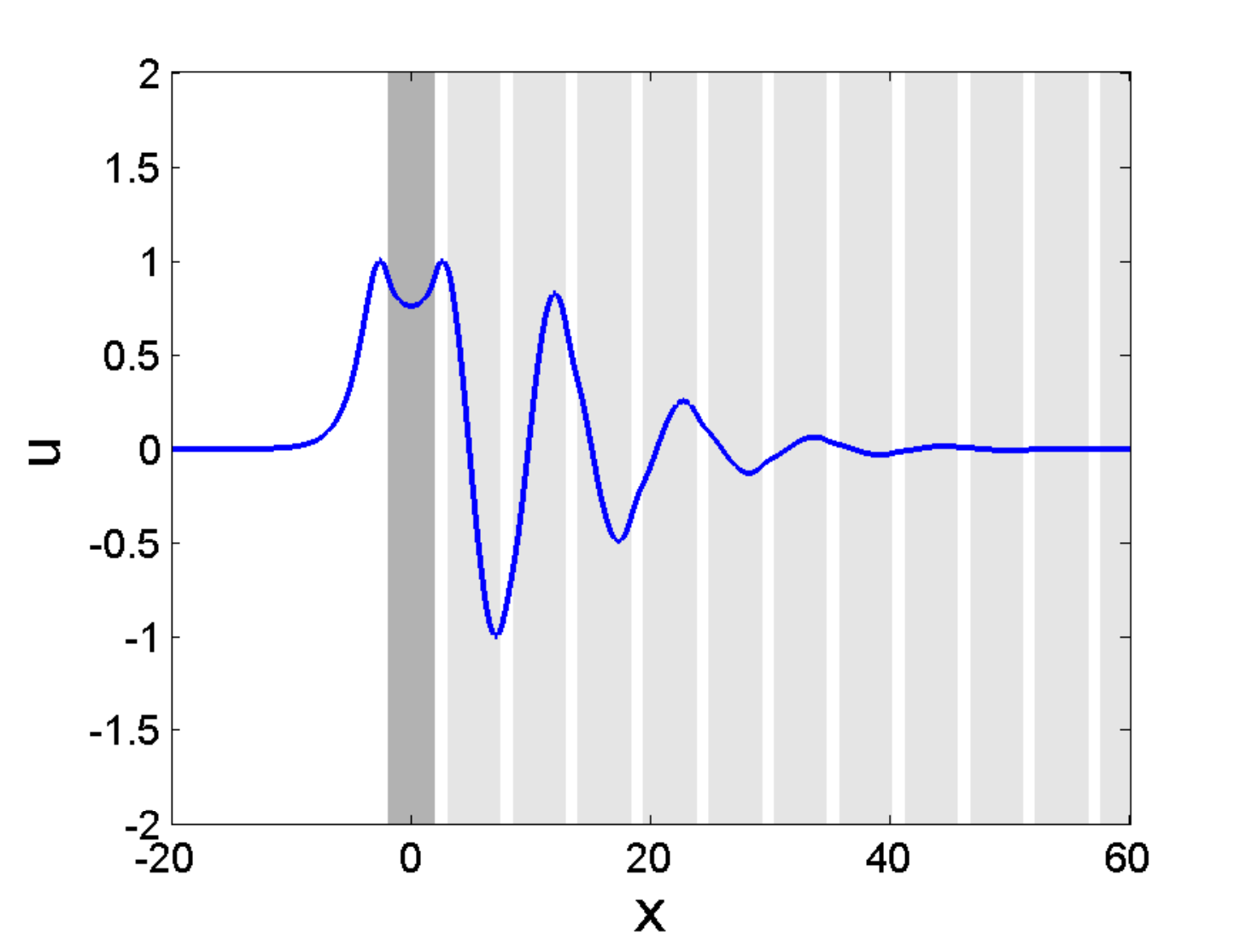}}
\subfigure[]{\includegraphics[scale=0.4]{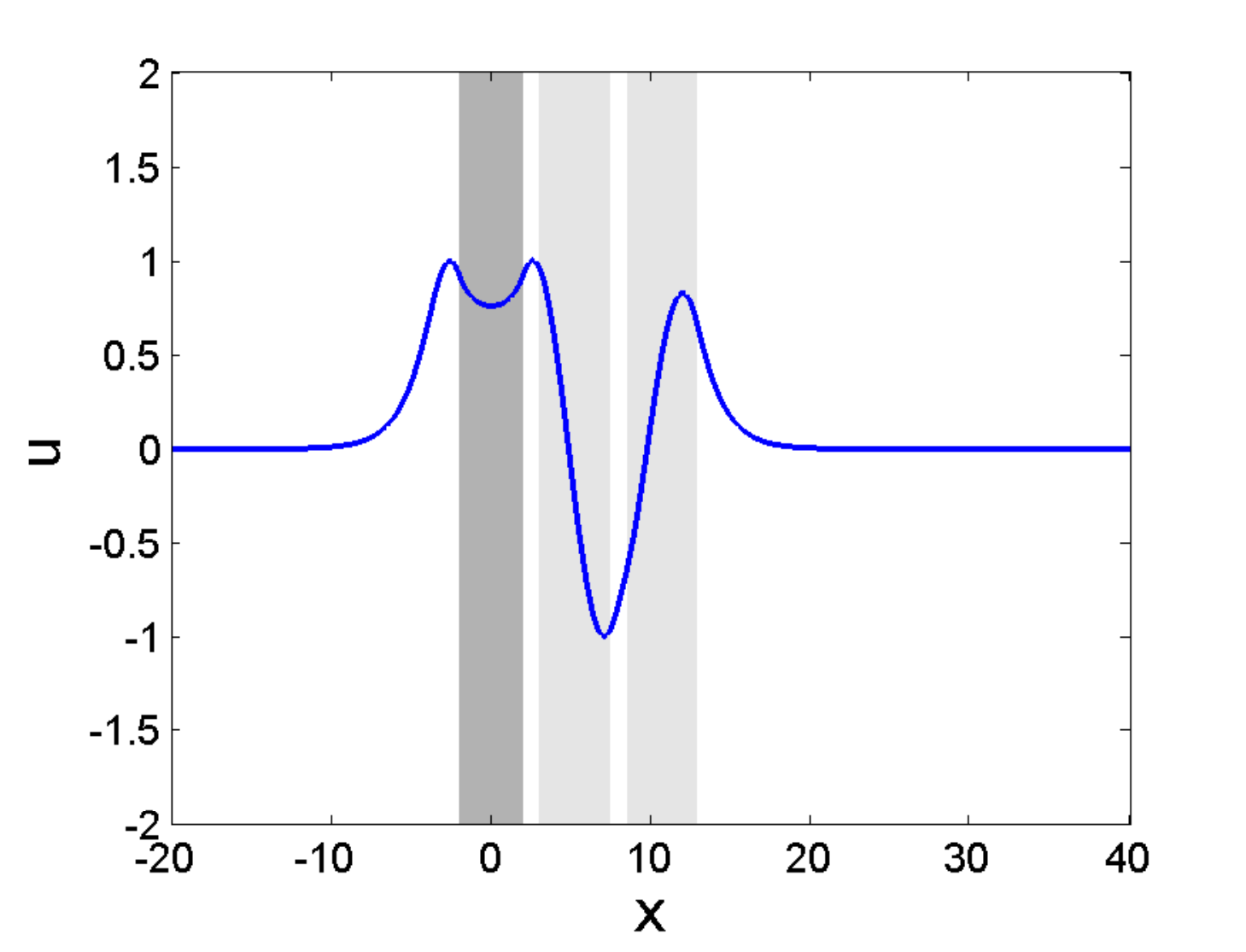}}
}
\caption{(a) A corresponding surface gap soliton of that in Figure \ref{fig:gapsoliton}. (b) An $f_1$ approximation of (a). }
\label{fig:surface}
\end{figure}

\begin{figure}[tbhp!]
\center{
\subfigure[]{\includegraphics[scale=0.6]{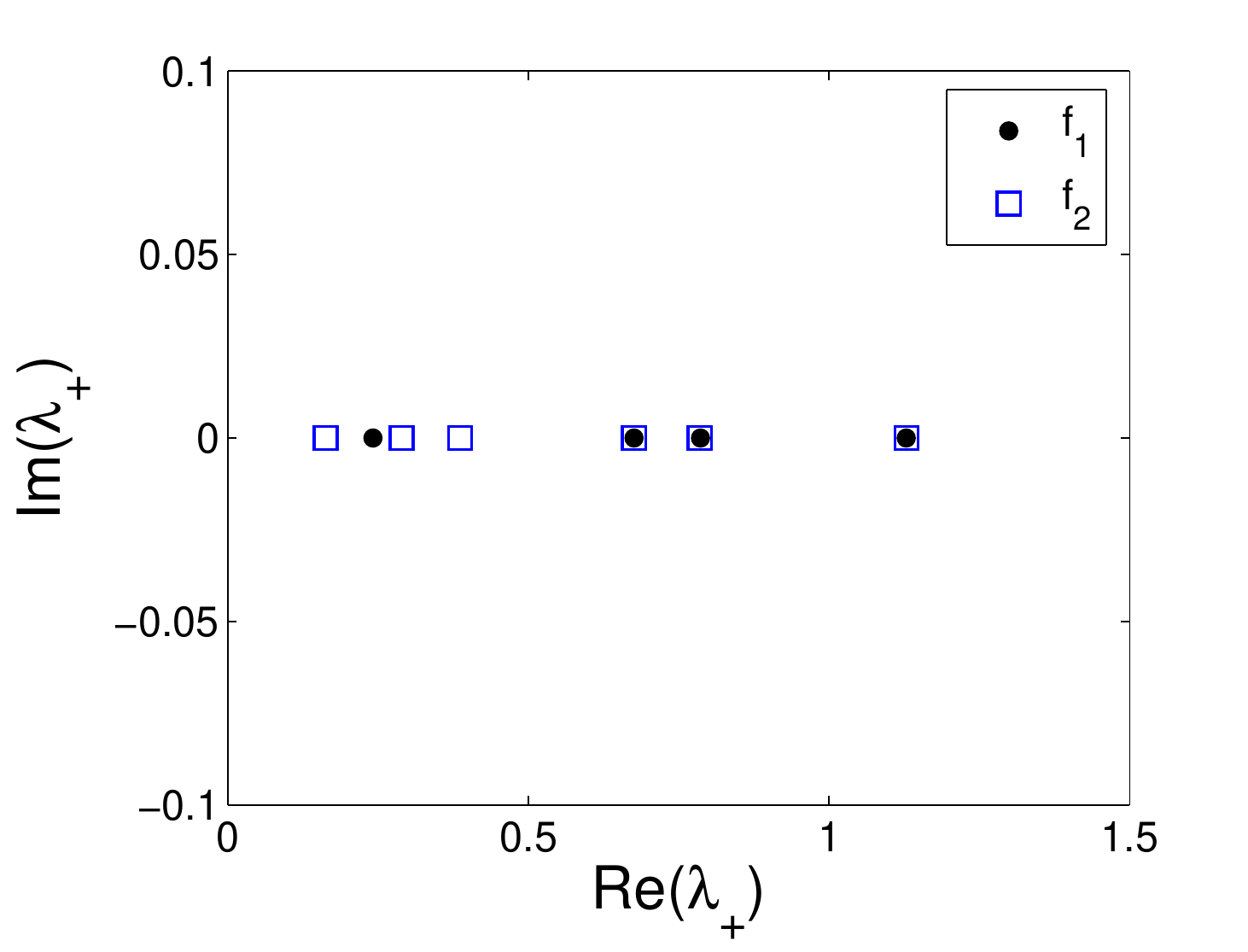}}
\subfigure[]{\includegraphics[scale=0.6]{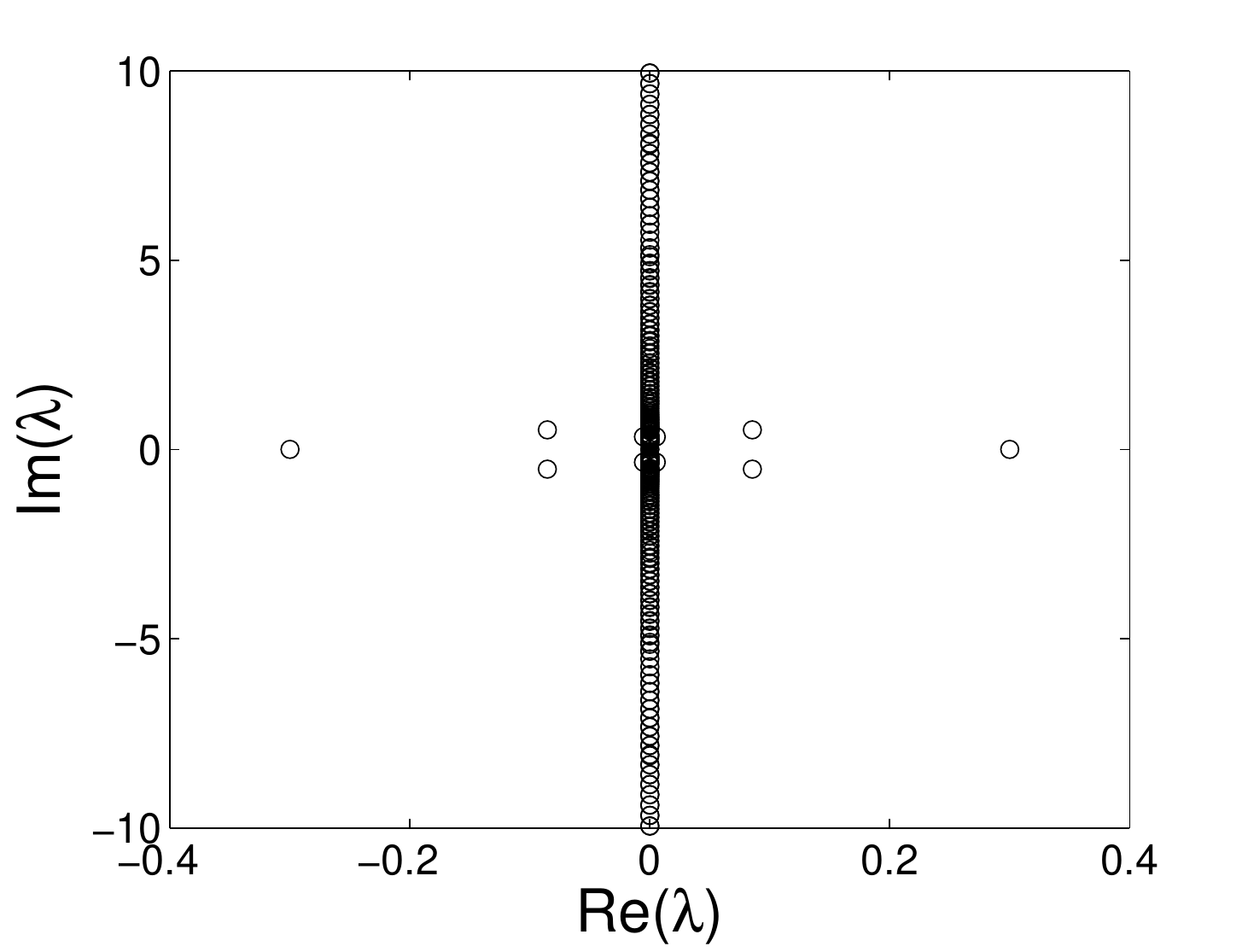}}
}
\caption{(a) The positive eigenvalues of $D_+$ for the approximations $f_1$ and $f_2$. Note that different from the plot in Figure \ref{fig:lambda_p}, here each symbol corresponds to one eigenvalue. (b) The eigenvalue structure of the surface gap soliton in Figure \ref{fig:surface}(a).}
\label{fig:surface_stability}
\end{figure}

First, we study Equation (\ref{gov1}) with
\begin{equation}
\eta=\left\{
\begin{array}{lll}
1,\quad x\in (-x_0,x_0),\\
0,\quad x\in (x_{2n+1},x_{2n+2}),\,(-x_{2n+2},-x_{2n+1}),
\end{array}
\right.
\label{eta1}
\end{equation}
where $x_0=2,\,x_{2n+1}-x_{2n}=1,\,x_{2n+2}-x_{2n+1}=\pi/\sqrt{\omega}$ and $ n=0,1,2,\dots$. A gap soliton for the above periodic inhomogeneity is depicted in Figure \ref{fig:gapsoliton}.

\begin{figure}[tbhp!]
\center{
\includegraphics[scale=0.6]{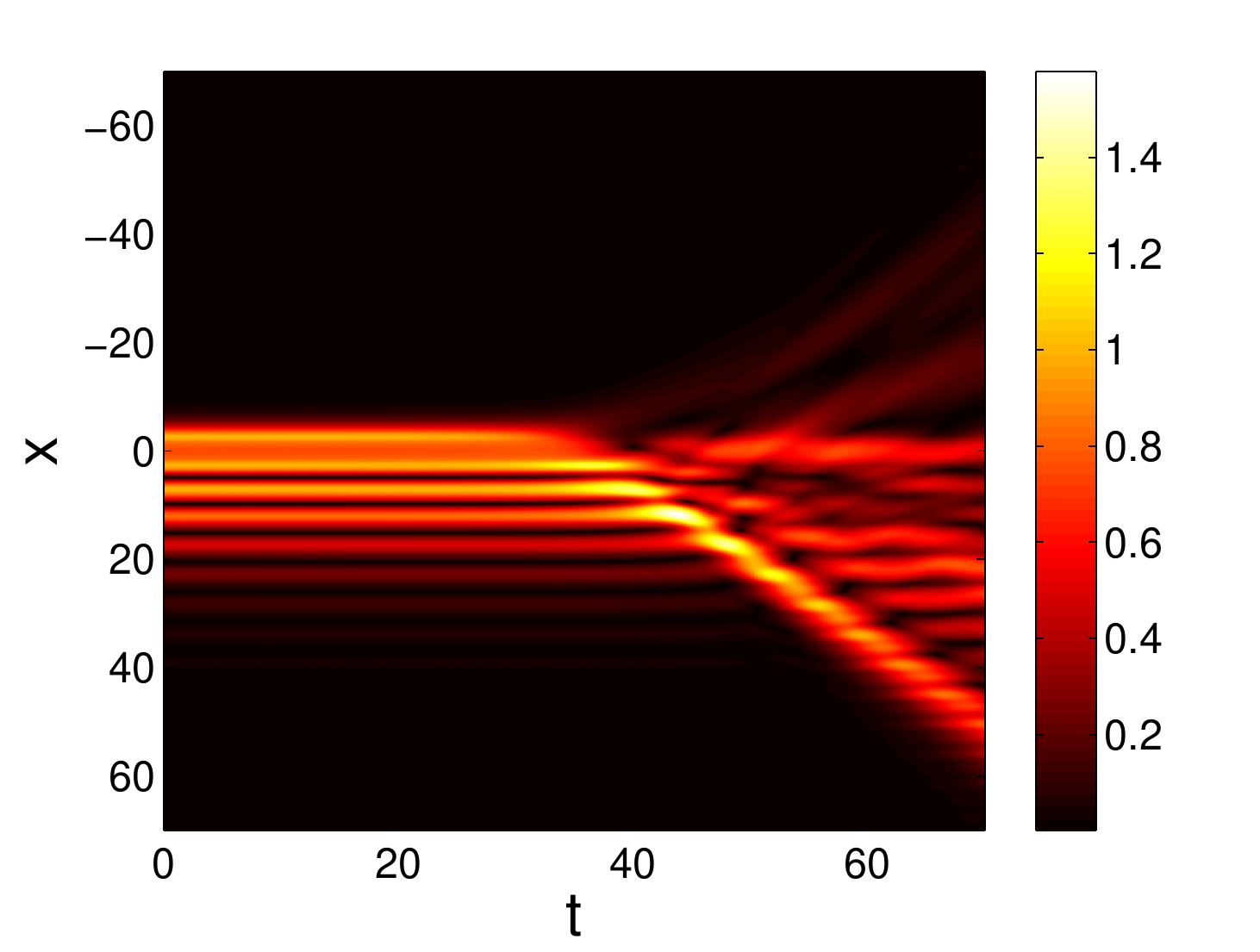}}
\caption{A time dynamics of the surface gap soliton in Figure \ref{fig:surface}. Shown is the top view of $|\psi(x,t)|$ in the $(t,x)$-plane.}
\label{fig:surface_evol}
\end{figure}

Theorem \ref{th:main} implies that to determine the instability of the gap soliton, it suffices to determine the instability of the corresponding solution $f_0$ shown in panel (a,b) of Figure \ref{fig:successivegapsolitons}. As discussed in \cite{rmckrtjhs10}, the positive solution $f_0$ is unstable, with $P=2$ and $Q=0$. We plot $\lambda_+$, i.e.\ the eigenvalues of the operator $D_+$, in Figure \ref{fig:lambda_p}. As shown in the figure, for $f_0$ there are two positive eigenvalues of $D_+$, i.e.\ $P=2$. The matrix $M$ in (\ref{eq:linear}) for the solution has one pair of real eigenvalues \cite{rmckrtjhs10} in agreement with Theorem \ref{th:ckrtj88}.

According to Lemma \ref{lem:main}, $f_n$ must have the same value of $P-Q$ as $f_0$. In the same figure, we obtain that $f_1$ and $f_2$ respectively has $P=6$ and $P=10$. Considering the fact from Figure \ref{fig:successivegapsolitons} that $f_1$ and $f_2$ respectively has $Q=4$ and $Q=8$, we indeed obtain that $P-Q=2$ for both $f_1$ and $f_2$. Using the lemma, one will obtain that $P-Q=2$ for $\lim_{n\to\infty}f_n$. Using Theorem \ref{th:main}, one can conclude that the gap soliton in Figure \ref{fig:gapsoliton} will be unstable. We depict in Figure \ref{fig:lambda}(a) the eigenvalue structure of the gap soliton in the complex plane. When the corresponding $f_0$ of the gap soliton has one pair of real eigenvalues \cite{rmckrtjhs10}, the gap soliton has several pairs of unstable eigenvalues. Nonetheless, one can easily notice that there is only one pair of real eigenvalues, similarly to $f_0$ \cite{rmckrtjhs10}. The time dynamics of the solution is shown in panel (b) of the same figure, where a typical instability is in the form of the dissociation of the solution.

Next, we study Equation (\ref{gov1}) with
\begin{equation}
\eta=\left\{
\begin{array}{lll}
1,\quad x\in (-x_0,x_0),\\
0,\quad x\in (x_{2n+1},x_{2n+2}),
\end{array}
\right.
\label{eta2}
\end{equation}
for the same values of $x_n$, $n=0,1,2,\dots$, as above. The only difference with $\eta$ defined in Equation (\ref{eta1}) is that the present periodic inhomogeneity only occupies the $x>0$-region. In this case, we will have surface gap solitons sitting at the interface between a homogeneous and a periodically inhomogeneous region. A corresponding surface gap soliton of that in Figure \ref{fig:gapsoliton} and one of its successive approximations $f_1$ are shown in Figure \ref{fig:surface}. The $f_0$ approximation of the soliton is nothing else but that shown in Figure \ref{fig:successivegapsolitons}(a).

Using Theorem \ref{th:main} and Remark \ref{remark}, one can expect that in this case $P-Q=2$. Plotted in Figure \ref{fig:surface_stability}(a) is the positive eigenvalues of $D_+$, i.e.\ $\lambda_+$. The positive eigenvalue $\lambda_+$ of $f_0$ is the same as before, which is $P=2$. For $f_1$ and $f_2$, from Figure \ref{fig:surface_stability}(a) one can deduce that $P=4$ and $P=6$, respectively, with $Q=2$ and $Q=4$. Hence, the limiting quantity $P-Q$ of the surface gap soliton is the same as that of the gap soliton in Figure \ref{fig:gapsoliton}, i.e.\ $P-Q=2$. As expected, shown in Figure \ref{fig:surface_stability}(b) is the eigenvalue structure of the gap soliton, where one also obtains one pair of real eigenvalues similarly to the stability the gap soliton depicted in Figure \ref{fig:lambda}(a). We plot the time dynamics of the surface gap soliton in Figure \ref{fig:surface_evol}.

\section{Conclusion}

We have considered a nonlinear Schr\"odinger equation with periodic inhomogeneity, both in the infinite and semi-infinite domain. Specifically we have studied the instability of gap solitons admitted by the system. We have established a proof that if the periodic inhomogeneity is arranged in a particular way, such that parts of the solutions belonging to closed trajectories in the phase-space have length half the period of the trajectories, then the solitons inherits the instability of the corresponding solution with finite inhomogeneity.  The analytical study is based on the application of a topological argument developed in \cite{ckrtj88}.

It is natural to extend the study to the case when the solutions are localized, but do not tend to the uniform zero solution (see, e.g., \cite{komi06_2}). The (in)stability of such solitons is proposed to be studied in the future using analytical methods similar to that presented herein.

\bibliographystyle{elsarticle-num}

\begin{thebibliography}{999}
\bibitem{abbond01} A. Abbondandolo. \emph{Morse Theory for Hamiltonian Systems.} Pitman Research Notes in Mathematics, vol. 425, Chapman and Hall, London, 2001.

\bibitem{acev00} A.B. Aceves, \emph{Optical gap solitons: Past, present, and future; theory and experiments}, Chaos 10, 584 (2000).

\bibitem{denz09} C. Denz, S. Flach, Yu.S. Kivshar, \emph{Nonlinearities in Periodic Structures and Metamaterials}, Volume 150 (Springer, 2009).

\bibitem{bara98} I. V. Barashenkov, D. E. Pelinovsky, and E. V. Zemlyanaya, \emph{Vibrations and Oscillatory Instabilities of Gap Solitons}, Phys. Rev. Lett. 80, 5117 (1998).

\bibitem{blan11} E. Blank and T. Dohnal, \emph{Families of Surface Gap Solitons and their Stability via the Numerical Evans Function Method}, to appear in SIAM J. Appl. Dyn. Syst..

\bibitem{chen87} W.\ Chen and D. L. Mills, \emph{Gap solitons and the nonlinear optical response of superlattices}, Phys. Rev. Lett. 58, 160 (1987)

\bibitem{ross98} A.\ De Rossi, C.\ Conti, and S.\ Trillo, \emph{Stability, Multistability, and Wobbling of Optical Gap Solitons}, Phys. Rev. Lett. 81, 85 (1998).

\bibitem{ster94} C. M. de Sterke and J. E. Sipe, in Progress in Optics, edited by E. Wolf (North-Holland, Amsterdam, 1994), Vol. XXXIII, pp. 203260.

\bibitem{dohn08} T. Dohnal and D. Pelinovsky, \emph{Surface Gap Solitons at a Nonlinearity Interface}, SIAM J. Appl. Dyn. Syst. 7, 249-264 (2008).

\bibitem{eggl96} B.J. Eggleton, R. E. Slusher, C. M. de Sterke, P.A. Krug, and J. E. Sipe, \emph{Bragg Grating Solitons}, Phys. Rev. Lett. 76, 1627 (1996).

\bibitem{eier04} B. Eiermann, Th. Anker, M. Albiez, M. Taglieber, P. Treutlein, K.-P. Marzlin, and M. K. Oberthaler, \emph{Bright Bose-Einstein Gap Solitons of Atoms with Repulsive Interaction}, Phys. Rev. Lett. 92, 230401 (2004).

\bibitem{good01} R.H. Goodman, M.I. Weinstein, and P.J.\ Holmes, \emph{Nonlinear propagation of light in one-dimensional periodic structures}, J.\ Nonlinear Science 11, 123168 (2001).

\bibitem{iizu94} T. Iizuka, \emph{Envelope Soliton of the Bloch Wave in Nonlinear Periodic Systems}, J. Phys. Soc. Jpn. 63, 4343 (1994).

\bibitem{iizu97} T. Iizuka and M. Wadati, \emph{Grating Solitons in Optical Fiber}, J. Phys. Soc. Jpn. 66, 2308 (1997) 

\bibitem{ckrtj88}
C.~K. R.~T. Jones, \emph{Instability of standing waves for non-linear
  schr\"odinger-type equations}, Ergodic Theory and Dynamical Systems
  \textbf{8*} (1988), 119--138.
  
\bibitem{rmckrtjhs10}
C.~K. R.~T. Jones, R.~Marangell, and H.~Susanto, \emph{Localized standing waves
  in inhomogeneous schr\"{o}dinger equations}, Nonlinearity \textbf{23} (2010),
  no.~2059.

\bibitem{kevr09} P.G.\ Kevrekidis, \emph{The discrete nonlinear Schr\"odinger equation: mathematical analysis, numerical computations and physical perspectives}, Volume 232 (Springer, 2009).

\bibitem{kevr08} P.G. Kevrekidis, D.J. Frantzeskakis, R.\ Carretero-Gonz\'alez (Eds.), \emph{Emergent nonlinear phenomena in Bose-Einstein condensates: theory and experiment}, Volume 45 (Springer, 2008).

\bibitem{komi06} Y. Kominis, \emph{Analytical solitary wave solutions of the nonlinear Kronig--Penney model in photonic structures}, Phys. Rev. E 73, 066619 (2006).

\bibitem{komi06_2} Y. Kominis and K. Hizanidis, \emph{Lattice solitons in self-defocusing optical media: analytical solutions of the nonlinear Kronig--Penney model}, Opt. Lett. 31, 2888-2890 (2006).

\bibitem{komi07} Y.\ Kominis, A.\ Papadopoulos, and K.\ Hizanidis, \emph{Surface solitons in waveguide arrays: Analytical solutions}, Opt. Express 15, 10041-10051 (2007).

\bibitem{malo94} B.A. Malomed and R.S. Tasgal, \emph{Vibration modes of a gap soliton in a nonlinear optical medium}, Phys. Rev. E 49, 5787 (1994).

\bibitem{mand03} D. Mandelik, H. S. Eisenberg, Y. Silberberg, R. Morandotti, and J. S. Aitchison, \emph{Band-Gap Structure of Waveguide Arrays and Excitation of Floquet-Bloch Solitons}, Phys. Rev. Lett. 90, 053902 (2003).

\bibitem{mand04} D. Mandelik, R. Morandotti, J. S. Aitchison, and Y. Silberberg, \emph{Gap Solitons in Waveguide Arrays},
Phys. Rev. Lett. 92, 093904 (2004).

\bibitem{peli05} D.E. Pelinovsky, P.G. Kevrekidis and D.J. Frantzeskakis, \emph{Stability of discrete solitons in Nonlinear Schrodinger Lattices}, Physica D 212, 1-19 (2005).

\bibitem{peli07} D.\ Pelinovsky and G.\ Schneider, \emph{Justification of the coupled-mode approximation for a nonlinear elliptic problem with a periodic potential}, Applicable Analysis 86, 10171036 (2007).

\bibitem{peli08} D.\ Pelinovsky and G.\ Schneider, \emph{Moving gap solitons in periodic potentials}, Mathematical Methods in the Applied Sciences 31, 17391760 (2008).

\bibitem{peli04} D.E. Pelinovsky, A.A. Sukhorukov, and Yu.S. Kivshar, \emph{Bifurcations and stability of gap solitons in periodic potentials}, Phys. Rev. E 70, 036618 (2004).

\bibitem{perko01}
L.~Perko, \emph{Differential equations and dynamical systems}, 3rd ed., Texts
  in Applied Mathematics, no.~7, Springer, 2001.

\bibitem{robsal93} J. Robbin, and D. Salamon. \emph{The Maslov index for paths.}Topology. Volume 32 Number 4. pp 827--844 (1993). 

\bibitem{rosb06} C.R. Rosberg, D.N. Neshev, W. Krolikowski, A. Mitchell, R.A. Vicencio, M. I. Molina, and Yu. S. Kivshar, \emph{Observation of Surface Gap Solitons in Semi-Infinite Waveguide Arrays}, Phys. Rev. Lett. 97, 083901 (2006)
    
\bibitem{smir06} E. Smirnov, M. Stepic, C. E. Ruter, D. Kip, and V. Shandarov, \emph{Observation of staggered surface solitary waves in one-dimensional waveguide arrays}, Opt. Lett. 31, 2338-2340 (2006).

\bibitem{volo81} Yu. V. Volovshchenko, Yu. N. Ryzhov, and V. E. Sotin, Zh. Tekh. Fiz. 51, 902 (1981) (in Russian) [Sov. Tech. Phys. Lett. 26, 541 (1981)].

\bibitem{yang10} J.\ Yang, \emph{Nonlinear Waves in Integrable and Nonintegrable Systems} (SIAM, 2010).

\end{thebibliography}

\end{document}